\documentclass[10pt, conference]{ieeeconf} 

\IEEEoverridecommandlockouts 

\usepackage{cite}
\usepackage{amsmath,amssymb,amsfonts}
\usepackage{algorithmic}
\usepackage{graphicx}
\usepackage{textcomp}
\usepackage{subfigure}
\usepackage[ruled,vlined,norelsize]{algorithm2e}
\usepackage[justification=centering]{caption}
\usepackage{color}

\newtheorem{remark}{Remark}
\newtheorem{definition}{Definition}
\newtheorem{problem}{Problem}

\newtheorem{eg}{Example}
\newtheorem{lemma}{Lemma}

\newtheorem{theorem}{Theorem}

\graphicspath{{./figure/}}

\title{\LARGE \bf
Opacity Enforcement by Edit Functions Under Incomparable Observations
}

\author{% <-this % stops a space
	Wei Duan,
	Ruotian Liu,
	Maria Pia Fanti,
	Christoforos N. Hadjicostis,
	and Zhiwu Li
    \thanks{Wei Duan, Ruotian Liu and Maria Pia Fanti are with the Department of Electrical and Information Engineering, Polytechnic University of Bari, Bari 71025, Italy (emails: \{wei.duan, ruotian.liu, mariapia.fanti\}@poliba.it).}  
	\thanks{Christoforos N. Hadjicostis is with the Department of Electrical and Computer Engineering, University of Cyprus, Nicosia 1678, Cyprus (e-mail: hadjicostis.christoforos@ucy.ac.cy).}
	\thanks{Zhiwu Li is with the School of Electro-Mechanical Engineering, Xidian University, Xi’an 710071, China, and also with the Institute of Systems Engineering, Macau University of Science and Technology, Taipa 999078, Macau SAR, China (e-mail: zhwli@xidian.edu.cn).}
}

\begin{document}
%\title{Opacity Enforcement by Edit Functions Under Incomparable Observations}

% \author{% <-this % stops a space
% 	Wei Duan,
% 	Ruotian Liu,
% 	Maria Pia Fanti,~\IEEEmembership{Fellow, IEEE},
% 	Christoforos N. Hadjicostis,~\IEEEmembership{Fellow, IEEE},\\
% 	and Zhiwu Li,~\IEEEmembership{Fellow, IEEE}	
% 	\thanks{*This work was supported in part by the National R\&D Program of China under Grant No. 2018YFB1700104, and the National Natural Science Foundation of China under Grants 61603285 and 61873342.} %($Corresponding$ $author$: $Chunfu$ $Zhong$).}% <-this % stops a space
% 	\thanks{Wei Duan is with the School of Electro-Mechanical Engineering, Xidian University, Xi'an 710071, China (e-mail: weiduan1103@outlook.com).}
% 	%
%     \thanks{Ruotian Liu and Maria Pia Fanti are with the Department of Electrical and Information Engineering, Polytechnic University of Bari, Bari 71025, Italy (emails: \{ruotian.liu, mariapia.fanti\}@poliba.it).}  
% 	%
% 	\thanks{Christoforos N. Hadjicostis is with the Department of Electrical and Computer Engineering, University of Cyprus, Nicosia 1678, Cyprus (e-mail: hadjicostis.christoforos@ucy.ac.cy).}
% 	%	
% 	\thanks{Zhiwu Li is with the School of Electro-Mechanical Engineering, Xidian University, Xi’an 710071, China, and also with the Institute of Systems Engineering, Macau University of Science and Technology, Taipa 999078, Macau SAR, China (e-mail: zhwli@xidian.edu.cn).}
% }

%\maketitle
\maketitle
\thispagestyle{empty}
\pagestyle{empty}

\begin{abstract}
As an information-flow privacy property, opacity characterizes whether a malicious external observer (referred to as an intruder) is able to infer the secret behavior of a system.
This paper addresses the problem of opacity enforcement using edit functions in discrete event systems modeled by partially observed deterministic finite automata.
A defender uses the edit function as an interface at the output of a system to manipulate actual observations through insertion, substitution, and deletion operations so that the intruder will be prevented from inferring the secret behavior of the system.
Unlike existing work which usually assumes that the observation capabilities of the intruder and the defender are identical, we consider a more general setting where these two entities may observe \emph{incomparable} subsets of events generated by the system.
%In particular, we consider a more general setting where the observation capability (to obtain observations from the system) of the edit function is \emph{incomparable} to that of the intruder, i.e., they may observe incomparable sets of events from the system.
To characterize whether the defender has the ability to enforce opacity of the system under this setting, the notion of \emph{$ic$-enforceability} is introduced.
Then, the opacity enforcement problem is transformed to a two-player game, with imperfect information between the system and the defender, which can be used to determine a feasible decision-making strategy for the defender.
Within the game scheme, an \emph{edit mechanism} is constructed to enumerate all feasible edit actions following system behavior.
We further show that an $ic$-enforcing edit function (if one exists) can be synthesized from the edit mechanism to enforce opacity.
\end{abstract}


% \begin{IEEEkeywords}
% 	Discrete event system, opacity enforcement, edit function, game theory, incomparable observations.
% \end{IEEEkeywords}

\section{Introduction}\label{sec:introduction}
%Security and privacy problems are increasingly important in all aspects of our lives.
%Introduced in the computer science community, 
Opacity is an information-flow privacy property that characterizes whether confidential information (referred to as the secret behavior) of a system can be inferred by a malicious intruder.
Opacity analysis and enforcement have become active research topics in discrete event systems (DESs) \cite{cassandras2021introduction,2020Chrisbook}.
Specifically, the intruder is generally assumed to have full knowledge of a system’s structure and attempts to infer the secret behavior of the system by implementing passive attacks (i.e., eavesdropping at the output of the system). 
A system is said to be opaque if, for any secret behavior,
there exists at least another non-secret behavior that appears identical to the intruder (i.e., it generates the same sequence of observations); this implies that the intruder can never infer the system's secret behavior.

According to diverse representations of the secret behavior of a system, several notions of opacity have been correspondingly formalized in the literature, including language-based opacity \cite{2011Lin, 2007Badouel}, current-state opacity \cite{2007Chris}, initial-state opacity \cite{2008Chris}, initial-and-final-state opacity \cite{2013Wu}, and so on.
Note that the above notions are developed in the context of automata models.
Alternative mathematical models for DESs, e.g., Petri nets, have also been widely used to formulate the notions of opacity (see, for example, \cite{bryans2005modelling, tong2015verification, tong2016a, tong2016b, zhu2022online}).

When a system is not opaque, the problem of opacity enforcement arises and has been extensively addressed following two main types of approaches.
In chronological order, the use of supervisory control theory was first proposed to construct minimally restrictive opacity-enforcing supervisory controllers \cite{2010Dubreil, 2012Chrisa}.
The idea is to restrict the system's behavior by a well-designed supervisor that disables some controllable events in order to avoid a violation of opacity. 
In particular, Yin and Lafortune \cite{2015Yin} proposed a game theoretical approach to embed all feasible supervisors in a finite structure.
Tong $\emph{et al.}$ \cite{2018Tong} adopted a similar approach but focused on a more general setting that considers incomparable observations between the intruder and the supervisor.

Instead of restricting system behavior, obfuscation techniques \cite{2014Wu,2016Wu, 2018Wu, 2018Ji, 2019Ji, 2020Moha, 2022Li, 2023Liu} are provided as an alternative approach to confuse the intruder by manipulating observations generated by the system via an output interface of the system. 
Specifically, Wu and Lafortune \cite{2014Wu,2016Wu} designed an insertion function, which takes as input an observable sequence generated by a given system and outputs a modified sequence by inserting fictitious events before each actual observation. 
The authors of \cite{2018Wu} proposed edit functions that are similar to insertion functions, but more powerful as they can alter the output behavior of a system by inserting, deleting, or substituting events.
Ji $\emph{et al.}$ \cite{2018Ji}, \cite{2019Ji} extended the setting of insertion and edit functions, respectively, into a more general case where the insertion and edit functions may become known to the intruder. 

Note that the common assumption regarding observation capabilities in opacity enforcement problems using obfuscation techniques \cite{2014Wu,2016Wu, 2018Wu, 2018Ji, 2019Ji, 2020Moha, 2022Li, 2023Liu} is that both the intruder and the insert/edit function possess partial but identical observation capabilities of the system.
In our prior work \cite{2023Duan}, we considered a scenario in which the intruder and the edit function obtain asymmetric information from the system, e.g., the edit function observes a subset of the events observed by the intruder. 
However, it is imperative to acknowledge that such an asymmetric information scenario is often insufficient to model real-world complexities effectively.
For instance, the issue of information confrontation always exists in network systems \cite{mitchell2015modeling, 2018GAN}, where both sides of the confrontation (defenders and intruders) may capture different information from a given system.
A more comprehensive representation of this issue would be that the defenders and intruders have imperfect information, which refers to the lack of complete or precise data during the decision-making processes of either side.
%In other words, although they can accurately understand the network status and topology, they may lack full capability of the each other's knowledge, e.g., they cannot distinguish exact actions undertaken by their respective opponents.

To characterize the nature of imperfect information obtained by the intruder and the defender on the problem of opacity enforcement via edit functions,\footnote{Unlike \cite{2018Ji}, \cite{2019Ji}, we only consider the case where the edit function is unknown to the intruder for the sake of simplicity.} in this paper we consider a more general setting (illustrated in Fig. \ref{illustration}).
We refer to this setting as the \emph{incomparable observation setting},\footnote{This is relevant to the verification of opaque systems and the evaluation of security and privacy aspects, as it captures the intuitive perception that the intruder and the defender may have partial and incomparable observation capabilities of a system.} under the assumption that the observations of a system captured by the intruder are incomparable to those obtained by the defender (edit function), i.e., their observable event sets do not necessarily have any inclusion relation.
In such case, the intruder and the edit function have imperfect information of the system since their estimates of the system obtained from their incomparable observations are not necessarily the same.

\begin{figure}[htbp]
	\centering
	\includegraphics[width=0.45\textwidth]{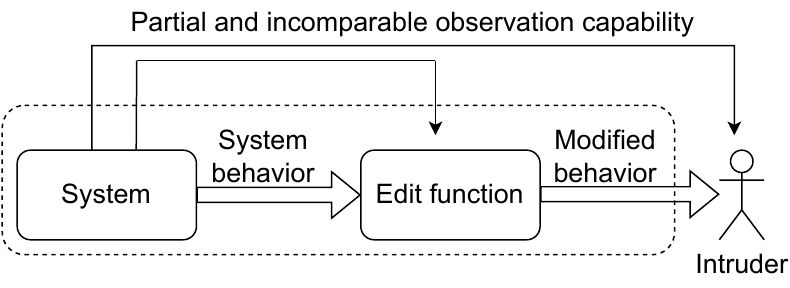}
	\caption{The edit mechanism.}
	\label{illustration}
\end{figure}

To obtain edit mechanisms under incomparable observations, we introduce the property of $ic$-enforceability that characterizes whether an edit function has the ability to enforce opacity under the incomparable observation setting.
Specifically, we construct an edit mechanism in a two-player game scheme with imperfect information between the system and the edit function.
In the end, we show that the edit function is $ic$-enforcing if it can be  synthesized from the edit mechanism. 

The remainder of this paper is organized as follows. 
Section~\ref{sec:preliminaries} presents the system model, the notion of current-state opacity, and the definition of the edit function. 
In Section~\ref{sec:formulation}, we introduce the incomparable observation setting, propose the $ic$-enforceability property, and formally outline the problem.
Section~\ref{sec:game} constructs a two-player game structure under the incomparable observation setting and identifies infeasible edit actions by proposing an appropriate utility function.
Then, Section~\ref{sec:mechanism} describes the procedure of building an edit mechanism and establishes a necessary and sufficient condition for edit functions to be $ic$-enforcing.
%Section~\ref{sec:editautomaton} develops an algorithm to synthesize $iu$-enforcing and $ik$-enforcing edit functions based on the corresponding edit mechanisms.
Section~\ref{sec:discussion} analyzes the complexity of the edit mechanism and engages in a comparative analysis of the proposed approach against existing methodologies.
Finally, conclusions along with future directions are reached in Section~\ref{sec:conclusion}.

\section{Preliminaries}
\label{sec:preliminaries}
Let $E$ be a finite set of events and $E^*$ denote the set of all finite length strings (finite sequences of events) over $E$, including the empty string $\varepsilon$. The length of a string $\lambda$ is the number of events in $\lambda$, denoted by $|\lambda|$, and the length of the empty string is denoted by $|\varepsilon|=0$. A language $L\subseteq E^*$ is a subset of finite-length strings. 
Given strings $u$ and $v$, $uv$ stands for the concatenation of $u$ and $v$, i.e., the sequence of events in $u$ followed by the sequence of events in $v$.

\subsection{System Model}
In this paper, we model a DES with a deterministic\footnote{We focus on deterministic finite automata to keep notation simple, but the results in this paper can be extended to nondeterministic finite automata in a straightforward manner.} finite automaton (DFA) $G =(X, E, f, x_0)$, where $X$ is a set of states, $E$ is a set of events, $f: X\times E \rightarrow X$ is the partial transition function, and $x_0\in X$ is the initial state.
The transition function $f$ can be extended to $X\times E^*\rightarrow X$ in the usual manner: $f(x,\varepsilon)= x$  and $f(x,\alpha \lambda)=f(f(x,\alpha),\lambda)$ for $x\in X$, $\alpha\in E$, and $\lambda\in E^*$.
Note that $f(x,\alpha \lambda)$ is undefined if $f(x,\alpha)$ is undefined.
The generated language of $G$, denoted by $L(G)$, is defined as $L(G) =\lbrace \lambda\in E^*\mid f(x_0,\lambda)$ is defined$\rbrace$.

The event set $E$ in a DFA is typically partitioned into the set of observable events $E_o$ and the set of unobservable events $E_{uo}=E\setminus E_o$, to capture partial observation of the DFA.
By taking into account the set of observable events, we define the system projection as $P: E^*\rightarrow E_o^*$, where $P(\varepsilon)=\varepsilon$, $P(\alpha)=\alpha$ if $\alpha\in E_o$, $P(\alpha)=\varepsilon$ if $\alpha\in E_{uo}$, and $P(\lambda \alpha)=P(\lambda)P(\alpha)$ for all $\lambda\in E^*$ and $\alpha\in E$.

\subsection{Current-State Opacity}

In the framework of DESs, two main categories of opacity (i.e., state-based and language-based opacity) have been widely studied in the literature \cite{2011Lin, 2007Badouel, 2007Chris, 2008Chris, 2013Wu, bryans2005modelling, tong2015verification}. 
Since the work in \cite{2013Wu} shows that these notions are transferable from one to another, in the remainder of the paper we focus on current-state opacity \cite{2007Chris}, which characterizes the secret behavior of the system by a subset of states $X_S\subseteq X$.
%A system is CSO if the intruder can never infer, from its observations, whether the current state of the system is a secret state or not.
In addition, we assume that the intruder has full knowledge of the system's structure and partially observes the system through its own observation capability, i.e., it can capture a subset of  observable events $E_I\subseteq E_o$.

By taking into account the set of events observed by the intruder, we define the intruder projection as $P_I: E^*\rightarrow E_I^*$, where $P_I(\varepsilon)=\varepsilon$, $P_I(\beta)=\beta$ if $\beta\in E_I$, $P_I(\beta)=\varepsilon$ if $\beta\in E\setminus E_I$, and $P_I(\lambda\beta)=P_I(\lambda)P_I(\beta)$ for all $\lambda\in E^*, \beta\in E$.
The inverse projection of $P_I$, denoted by $P_I^{-1}: E_I^*\rightarrow 2^{E^*}$, is defined as $P_I^{-1}(\beta) = \lbrace \lambda\in E^* | P_I(\lambda) = \beta \rbrace$, where $\beta\in E_I^*$.

\begin{definition}\label{CSO1}
	Given a DFA $G=(X,E,f,x_0)$, the intruder projection $P_I$ and the set of secret states $X_S\subseteq X$, $G$ is said to be current-state opaque with respect to $P_I$ and $X_S$ if 
    for all $\lambda\in L(G)$ such that $f(x_0,\lambda)\in X_S$, there exists $\lambda'\in L(G)$ such that (i) $P_I(\lambda)=P_I(\lambda')$ and (ii) $f(x_0,\lambda')\in X\backslash X_S$.
\end{definition}

We use CSO to denote current-state opaque with respect to $P_I$ and $X_S$ for short.
A system is CSO if the intruder can never infer, based on its observations, whether the current state of the system is in a secret state or not.
%By Definition \ref{CSO1}, CSO requires that the evolution of the system to states in $X_S$ is kept uncertain to the intruder (under the intruder projection map $P_I$), at least until the current state of the system leaves $X_S$.
For the sake of simplicity, we denote the sublanguages that reach secret and non-secret states by $L_S=\lbrace \lambda\in L(G)|f(x_0,\lambda)\in X_S\rbrace$ and $L_{NS}=\lbrace \lambda\in L(G)|f(x_0,\lambda)\in X\backslash X_S\rbrace$, respectively.

%\begin{definition}\label{CSO}
%	Given a system $G$, the resilient hypothetical projection $P_{RH}$ and the set of secret states $X_S\subseteq X$, $G$ is said to be CSO with respect to $P_{RH}$ if $\forall t\in L_S, \exists t'\in  L_{NS}$ such that $P_{RH}(t)=P_{RH}(t')$.
%\end{definition}
%
%\begin{proposition}\label{ProCSO}
%	Given a system $G$, if CSO with respect to $P_{RH}$ holds, then CSO with respect to $P_I$ holds.
%\end{proposition}
%
%\begin{proof}
%	If $G$ is CSO with respect to $P_{RH}$, it holds for all $ t\in L_S$, there exists $t'\in  L_{NS}$ such that $P_{RH}(t)=P_{RH}(t')$. By the definition of $E_{RH}$, we know $E_I\subseteq E_{RH}$ and $E_I^* \subseteq E^*_{RH}$, which implies $P_{I}(t)=P_{I}(t')$. 
%	Thus, $G$ is CSO with respect to $P_I$. 
%\end{proof}
%
%\textcolor{blue}{(Will be deleted) Note that when $E_I\subseteq E_D$, the hypothetical set of events satisfies $E_{RH}=E_o$, i.e., the intruder is assumed to have the same observation capability as the edit function.
%Since this case can be solved in a manner similar to the techniques in \cite{2014Wu,2018Wu}, we consider the other case ($E_D\subseteq E_I$) in the remainder of this paper. }
%\textcolor{red}{(Changed to) We use CSO to denote CSO with respect to $P_{RH}$ for short.
%In the remainder of this paper, we mainly focus on the case $E_D\subseteq E_I$ since the other one can be solved in a manner similar to the techniques in \cite{2014Wu,2018Wu}.}

\vspace{-0.5em}
\subsection{Edit Function}

To enforce opacity, an obfuscation mechanism was proposed in \cite{2014Wu,2018Wu,2019Ji} via edit functions to manipulate actual observations generated by a system via insertion, substitution, and deletion operations of events.
The edit function is an interface placed at the output of the system and in this work we assume that the defender that implements an edit function may observe a subset of observable events of the system, namely, $E_D\subseteq E_o$. 
Note that the implementation of the edit function is assumed to be unknown to the intruder in this paper.

By taking into account the set of events observed by an edit function, the defender projection $P_D: E^*\rightarrow E_D^*$ is defined as $P_D(\varepsilon)=\varepsilon$, $P_D(\gamma)=\gamma$ if $\gamma\in E_D$, $P_D(\gamma)=\varepsilon$ if $\gamma\in E\setminus E_D$, and $P_D(\lambda\gamma)=P_D(\lambda)P_D(\gamma)$ for all $\lambda\in E^*, \gamma\in E$.
With a slight modification on the edit function in \cite{2018Wu}, an event-based version is defined (to make the exposition simpler) as follows.

 \begin{definition}(Edit function) \label{editfunction}
	An edit function is defined by $f_e: E_o^*\times E_o\rightarrow E_o^*$ such that for $\lambda\gamma\in L(G)$, $\lambda\in E^*, \gamma\in E$\\ i) for $\gamma\in E\setminus E_D$,  $f_e(\lambda,\gamma)=\gamma$;\\ ii) for $\gamma\in E_D$, it holds: 
	\begin{equation*}
	f_e(\lambda,\gamma)=\left\{
	\begin{aligned}
	& \gamma_I\gamma,  \ \gamma_I\in E_D^* \ \textit{\rm is inserted before} \ \gamma; \\ 
	& \gamma_R,  \ \gamma \ \textit{\rm is substituted with} \ \gamma_R\in E_D;\\ 
	& \varepsilon,  \ \gamma \ \textit{\rm is erased}.\\ 
	\end{aligned}
	\right. 
	\end{equation*}
\end{definition}

\vspace{0.5em}
% $f_e(\gamma)=\lbrace\gamma_I\gamma|\gamma_I\in E_D^*\rbrace$ if $\gamma_I$ is inserted before $\gamma$; $f_e(\gamma)=\lbrace\varepsilon\rbrace$ if $\gamma$ is erased; $f_e(\gamma)=\lbrace\gamma_R|\gamma_R\in E_D\rbrace$ if $\gamma$ is replaced with $\gamma_R$.}

Based on the edit function's own observations, one aims to manipulate the actual observations generated by a system so as to corrupt the observations received by the intruder.
Note that an edit function $f_e$ considered in this paper is deterministic: following each observation in $E_D$, it can choose one type of edited operations (insertion, substitution, or deletion) to modify the actual observations; importantly, observations in $E_o\setminus E_D$ cannot be edited since they cannot even be observed by the edit function.

% For notational convenience, we let the domain and codomain of $f_e$ to be $E_o$ and $E_o^*$ to allow an edit function to ``react'' to all observable events of a system.
% That is, $f_e(\gamma)=\gamma$ for $\gamma\in E_o\setminus E_D$ in the sense that the edit function cannot manipulate events that it does not observe, whereas $f_e(\gamma)$ can implement edited operations for $\gamma\in E_D$, i.e., it can output $\gamma_I\gamma$ by inserting $\gamma_I$ before $\gamma$, $\gamma_R$ by substituting $\gamma$ with $\gamma_R$, or $\varepsilon$ by deleting $\gamma$. 

In the following, we assume that $\gamma_I$ is of bounded length\footnote{By introducing an upper bound $K$ on the length of the inserted string, we limit the capabilities of the edit function. This ensures that no strings of unbounded length can be inserted during the insertion operations, providing control over the size and complexity of the edited output.} for the sake of simplicity, i.e., $K\in\mathbb{N}$ is the maximum number for insertions.
We use $E_D^{\leq (K+1)}\subset E_D^*$ to denote the finite set of strings of maximum length $K+1$ that can be chosen by an edit function, which represents all types of edited operations generated by the edit function.
Then, we use $E_o^{\leq (K+1)}\supseteq (E_D^{\leq (K+1)}\cup (E_o\setminus E_D))$ to denote the finite set of strings that can be chosen by an edit function (including the events that cannot observed by the edit function).
For simplicity, we let $E_o^{\leq (K+1)}$ be the codomain of the edit function to keep the notation simple (i.e., $f_e : E_o^*\times E_o\rightarrow E_o^{\leq (K+1)}$).

An edit function can be extended to a string-based version in a recursive manner as: $f_e(\varepsilon)=\varepsilon$ and $f_e(\sigma \gamma)=f_e(\sigma)f_e(\gamma)$ for $\sigma\in E_o^*$ and $\gamma\in E_o$. 
Note that the output of an edit function should be implemented in a causal manner, i.e., it is determined by the current observed event, the previous observations, and the previous edit actions.
For instance, given an observable sequence $\gamma = \gamma_1\gamma_2\in E_o^*$, the edit function $f_e(\gamma_2)$ not only relies on the current event $\gamma_2$, but is also affected by $\gamma_1$ and $f_e(\gamma_1)$. 
Moreover, the edit action $f_e(\gamma) = f_e(\gamma_1)f_e(\gamma_2)$ may not be defined in the original system if either $f_e(\gamma_1)$ or $f_e(\gamma_2)$ chooses an infeasible edit.

\section{Current-state opacity enforcement via edit functions under incomparable observations}
\label{sec:formulation}

\begin{figure*}
 \centering
 \subfigure[Observational setting in \cite{2014Wu, 2018Wu, 2018Ji, 2019Ji}]{
 \includegraphics[scale=1.2]{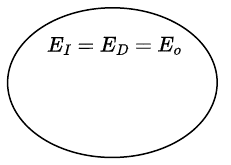}
 }
 \subfigure[Observational setting in \cite{2023Duan}]{
 \includegraphics[width=0.26\textwidth]{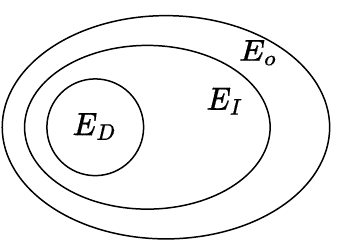}
 }
 \subfigure[Observational setting in this paper]{
 \includegraphics[width=0.32\textwidth]{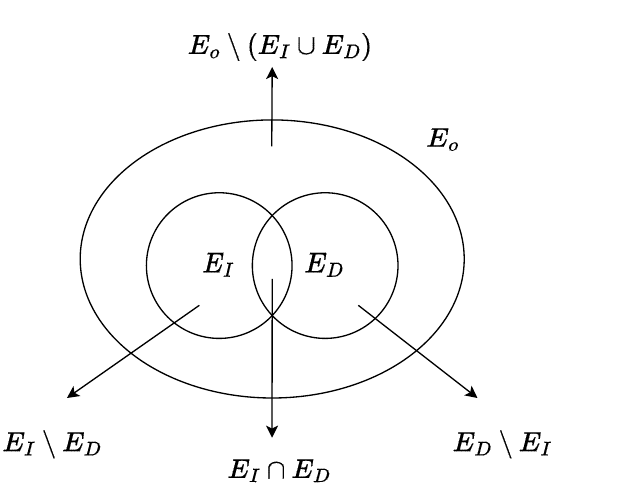}
 \label{setrelationcom}
 }
    \caption{Comparison of observation capability between the intruder (I) and the defender (D) in different works.}
    \label{fig:relation}
\end{figure*}

In this section, we first introduce a general setting by allowing the observations of the edit function for a given system to be \emph{incomparable} to those of the intruder.
Then, the notion of \emph{$ic$-enforceability} is proposed to characterize whether the edit function has the ability to enforce opacity.
In the end, the problem of current-state opacity enforcement via edit functions under incomparable observations is formulated. 

\subsection{Incomparable Observation Setting}

% Understanding the intruder's awareness of the defender in a security context is crucial because it can affect the effectiveness of security measures. An intruder who is aware of the defender may be able to avoid detection or tailor their attack to exploit the defender's vulnerabilities. In contrast, an intruder who is not aware of the defender may be more likely to make mistakes or leave evidence of their presence. Therefore, the intruder's awareness can impact the likelihood of a successful attack and the ease with which it can be carried out.  
In practical scenarios involving the enforcement of opacity for a given system, challenges arise due to the possibility that the intruder and the edit function (defender) might only have access to partial and distinct information extracted from the system.
For instance, the intruder or the edit function may only be able to capture observations generated by the system that are not necessarily the same.
%: (i) either the intruder or the edit function may not be able to capture the full observation of the system; (ii) the intruder may be aware of the existence of the edit function. 
To address such an issue, this paper considers the problem of CSO enforcement via edit functions under a more general setting where the observations captured by the intruder and the edit function are incomparable.

Given a DES modeled by a DFA $G=(X,E,f,x_0)$ with $E=E_o\cup E_{uo}$, the sets of events observed by the intruder and an edit function are $E_I\subseteq E_o$ and $E_D\subseteq E_o$, respectively. 
We introduce the \emph{incomparable observation setting} by considering the following assumptions:
\begin{enumerate}
	\item The intruder and the edit function have full knowledge of the system's structure.
    \item The edit function is aware of the existence of the intruder and its observation capability.
    %\footnote{We adopt this assumption for the sake of simplicity. A more challenging assumption would be that the edit function is unaware of the existence of the intruder, which is discussed later in the paper.}
	\item The intruder and the edit function have incomparable observation capabilities, i.e., $E_I$ and $E_D$ do not necessarily satisfy an inclusion relation.\footnote{For example, given a set of observable events $E_o=\lbrace a,b,c,d\rbrace$, we could have that $E_D=\lbrace a,b,c\rbrace$ and $E_I=\lbrace b,c,d\rbrace$.} 
\end{enumerate} 

Assumption 1) is standard and allows the intruder (the edit function) to attack (defend).
%setting that make the intruder and the edit function be available to attack (i.e., infer the occurrence of the secret states via unauthorized access) and defense (i.e., protect the security of the system by modifying observations perceived by the intruder).
Assumption 2) allows the edit function to know potential risks, e.g., if there exist vulnerable sensors in the system, the edit function can treat the signals outputted by these sensors as observations received by the intruder.
Assumption~3) indicates that the imperfect information of the system can be derived from the intruder and the edit function.
In such a case, the edit function cannot solely design feasible edit actions based on its own observations. 
It also faces the additional challenge of ensuring that the edit actions cannot be recognized by the intruder based on its observations.

% \footnote{A similar setting was investigated in \cite{2018Tong}, which considered incomparable observations between the intruder and the supervisor. However, our approach is more challenging as we address this issue using edit functions, where we also need to ensure that the intruder can recognize the outputs of the edit function.}

\begin{remark}
Note that a similar setting was investigated in \cite{2018Tong} through supervisory control under incomparable observations between the intruder and the supervisor, restricting the behavior of the system. In contrast, we address this issue using edit functions (i.e., by manipulating the observations of the system), which is more challenging since edits must be chosen not only to hide secrets but also to remain unambiguously recognized by the intruder.
Technically, the solution requires the formulation of a two-player game with imperfect information.
To our knowledge, the synthesis and solution of such imperfect-information games for security problems in DESs has not yet been explored.
\end{remark}

In order to highlight the challenges that might incur under the incomparable observation setting considered in this paper, we use the Venn diagrams in Fig.~\ref{fig:relation} to illustrate the logical relation between the sets of events observed by the intruder and the edit function (defender).
In contrast to the prior work \cite{2014Wu, 2018Wu, 2018Ji, 2019Ji, 2023Duan}, when the edit function tries to implement edit actions to react to system behavior, there are four sets (shown in Fig. \ref{setrelationcom}) that have to be taken into account due to the incomparable observation setting: (1) $E_o\setminus (E_I\cup E_D)$, (2) $E_I\setminus E_D$, (3) $E_D\setminus E_I$, and (4) $E_D\cap E_I$.
Specifically, when the system executes an event $\alpha$ that belongs to either the first event set or the second one (e.g., $\alpha\in E_o\setminus (E_I\cup E_D)$ or $\alpha\in E_I\setminus E_D$), the edit function can only directly ``output'' $\alpha$ (without any edit action); since it cannot even observe $\alpha$, the edit function ``outputs'' $\alpha$ (i.e., it does nothing).
On the contrary, when the system executes an event $\alpha$ that belongs to either the third event set or the fourth one (e.g., $\alpha\in E_D\setminus E_I$ or $\alpha\in E_D\cap E_I$), the edit function can react to $\alpha$ since it can observe it; however, at this point, the edit function can also output an edit action that is either observed by the intruder or not (e.g., the edit action belongs to either $E_D\cap E_I$ or $E_D\setminus E_I$, respectively). 
Thus, compared to the prior work in \cite{2014Wu, 2018Wu, 2018Ji, 2019Ji, 2023Duan}, the incomparability of observation capabilities requires a more complex strategy using edit functions to enforce opacity.
To better illustrate our motivations, a running example is given below.

\begin{eg}\label{example1}
	Consider the DFA $G=(X,E,f,x_0)$ in Fig.~\ref{motivation}, where the set of observable events is $E_o=\lbrace a,b,c,d\rbrace$ and the set of secret states is $X_S=\lbrace 5\rbrace$.

    \begin{figure}[htbp]
	\centering
	\includegraphics[width=0.3\textwidth]{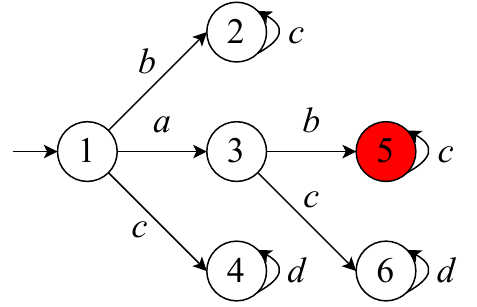}
	\caption{Motivating example.}
	\label{motivation}
    \end{figure}

    The intruder and an edit function have incomparable sets of observations, i.e., $E_I=\lbrace a,b,d\rbrace$ and $E_D=\lbrace b,c,d\rbrace$.
    Thus, the estimate of the system states by the intruder is not necessarily identical to that of the edit function.
    For example, given a sequence $abc\in L(G)$ generated by the system, the intruder can estimate that the system is in state $\lbrace 5\rbrace$ since it can observe $ab$ via $P_I(abc)=ab$, whereas the edit function cannot distinguish whether the system is in state $\lbrace 2\rbrace$ or $\lbrace 5\rbrace$ since it can only observe $bc$ via $P_D(abc)=bc$. 
    In this case,\footnote{On the contrary, if $E_D=\lbrace a,b,d\rbrace$ and $E_I=\lbrace b,c,d\rbrace$, the edit function is able to infer the secret state but the intruder cannot. At this point, the edit function does not need to implement any edit action since the intruder is already confused by the actual observation it received.} even though the edit function does not infer the secret state from its own observations $ab$, it is necessary to implement edit actions as the intruder is able to infer the secret state by its own observations $bc$.
    Under this incomparable observation setting, problems arise in terms of when and how to modify observations using the edit function in a allowable way.

    For simplicity and better illustration, consider the substitution operations acted by the edit function in the running example.
    That is, given the sequence $abc\in L(G)$ generated by the system, the edit function could either (i) output edited sequence $acd$ via $f_e(abc)=acd$ (by replacing $b$ with $c$ and replacing $c$ with $d$) or (ii) $abc$ via $f_e(abc)=abc$ without any edit action (note that the edit function can only perform edit actions on its received observations $bc$, i.e., $P_D(abc)=bc$).
    At this point, the intruder can estimate that the system is in state $\lbrace 6\rbrace$ if $f_e(abc)=acd$ due to observations $ad$ outputted by the edit function, i.e., $P_I(acd)=ad$; or it can estimate that the system is in state $\lbrace 5\rbrace$ if $f_e(abc)=abc$ due to observations $ab$ outputted by the edit function, i.e., $P_I(abc)=ab$.
    Apparently, the latter case is not allowable since the intruder is still able to infer the secret state after edit actions.
    $\hfill\diamond$
\end{eg}

% Our goal is to design two feasible edit function such that it protects the secret behavior of the system under the incomparable observation setting when the edit function is known and unknown.
% Such properties are characterized as ``$ik$-enforceability'' and ``$iu$-enforceability'' in the later subsection.

\vspace{0.5em}

\subsection{$IC$-Enforceability}

The notion of \emph{$ic$-enforceability} is introduced based on the specifications of availability, confidentiality, and integrity in order to characterize the ability of an edit function to confuse the intruder under the incomparable observation setting.

\begin{definition}($i$-availability)
    Consider a system $G$ with respect to the projection $P$. An edit function $f_e$ is $i$-available if for all $\sigma\in P[L(G)], f_e(\sigma)$ is defined.
\end{definition}

\begin{definition}($c$-availability)
    Consider a system $G$ with respect to the projections $P$ and $P_D$. An edit function $f_e$ is $c$-available if for all $\sigma,\sigma'\in P[L(G)]$, $ P_D(\sigma)=P_D(\sigma')$ implies that $f_e(\sigma)$ and $f_e(\sigma')$ are defined such that $P_D(f_e(\sigma))=P_D(f_e(\sigma'))$.
\end{definition}

Availability ensures that an edit function is accessible and usable when needed. 
Specifically, $i$-availability requires that an edit function should react to each observable event generated by the system (recall that, by construction, no modification is allowed if (i) the event is not observed by the edit function or (ii) the event is observed by the edit function but it is substituted with itself). 
Moreover, under the incomparable observation setting, an edit function should also ensure that its output can be recognized by the intruder.
Thus, $c$-availability is proposed to require that an edit function should implement the same edit actions to all sequences that have the same defender projection (the edit function cannot react differently since its observations are identical).
For the sake of simplicity, we say that an edit function is available if it is $i$-available and $c$-available.

\begin{definition}(Confidentiality)\label{defconfidential}
    Consider a system $G$ with respect to the projections $P$ and $P_I$. An edit function $f_e$ is confidential if for all $\sigma\in P(L_S)$, $P_I^{-1}[P_I(f_e(\sigma))]\cap L_{NS}\neq~\emptyset$.
\end{definition}

Confidentiality involves protecting the secret states from unauthorized access by requiring that, for each sequence reaching secret states, an edit function should generate an edited sequence such that the secret states cannot be revealed from the intruder’s estimates.

\begin{remark}
In prior works such as \cite{2014Wu, 2018Wu, 2018Ji, 2019Ji, 2023Duan}, enforcing opacity means that either the insertion function or the edit function must consistently output a sequence that reaches non-secret states, ensuring the intruder cannot deduce the current state of the system as a secret state.
To enhance the flexibility of the edit function, we require that the edit function should only confuse the intruder when necessary.
By taking advantages of the notion of confidentiality formulated in Definition \ref{defconfidential}, the edit function is also able to output a sequence that reaches secret states as long as the system is in non-secret states in reality; that is, the secret states are not reached in reality but we let the intruder believe that they are reached.  
$\hfill\diamond$
\end{remark}

\begin{definition}(Integrity)
    Consider a system $G$ with respect to the projection $P$. An edit function $f_e$ is integral if for all $\sigma \beta\in P[L(G)]$, $f_e(\sigma \beta)=f_e(\sigma)f_e(\beta)$ is available and confidential.
\end{definition}

Integrity refers to the ability of maintaining the accuracy and consistency of the edit actions. 
That is, an edit function should ensure that each subsequent edited sequence maintains availability and confidentiality.

\begin{definition}
    An edit function is $ic$-enforcing if it is available, confidential, and integral.
\end{definition}

\subsection{Problem formulation}

Consider a system modeled by a DFA $G=(X,E,f,x_0)$, where the set of secret states is $X_S\subseteq X$.
Assume that the intruder and an edit function have incomparable observation capabilities of the system.
Our goal is to design an edit mechanism such that any edit function synthesized from it is able to protect the secret states from being revealed by the intruder when the secret states are reached in reality.
The problem is formally presented as follows.
%In the remainder of the paper, we aim to develop two synthesis algorithms for $iu$-enforcing and $ik$-enforcing edit functions. 
 
\begin{problem}
    Consider a system modeled as DFA $G=(X,E,f,x_0)$ with a set of secret states $X_S\subseteq X$, the set of events observed by the intruder $E_I$, and the set of events observed by an edit function $E_D$, where $E_I$ and $E_D$ are incomparable. We aim to construct an edit mechanism for all edit functions to be $ic$-enforcing.
\end{problem}

	\graphicspath{{./figure/}}
	\begin{figure*}[htbp]
		\centering
  	\subfigure[System $\emph G$.]{
		\includegraphics[width=0.3\textwidth]{System.pdf}
			\label{system}
		}
		\subfigure[Intruder observer $\mathbf{O}_I$.]{
			\includegraphics[width=0.26\textwidth]{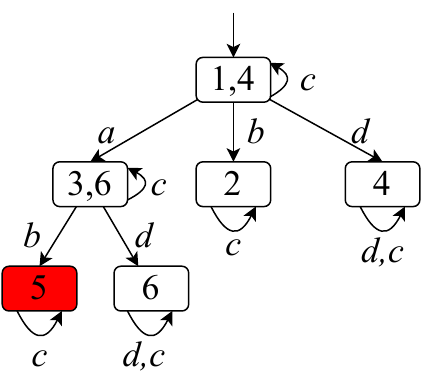}
			\label{DiagramOI}
		}
    \subfigure[Defender observer $\mathbf{O}_D$.]{
		\includegraphics[width=0.16\linewidth]{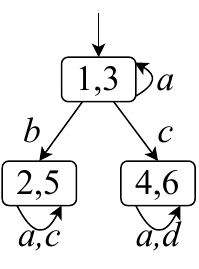}
			\label{DiagramOD}
		}	
		\caption{Illustrations of the observers.}
  \label{DiagramOIOD}
	\end{figure*}
 
% The underlying idea that how to construct an edit mechanism involves five steps (see Fig. \ref{Constructionflow}):
% \begin{enumerate}
%     \item Constructing observers to estimate the current state of a given system from the perspectives of the system, the intruder, and the defender;
%     \item Incorporating the system, intruder, and defender observers in a game manner to construct an edit game structure by assuming the intruder and the defender can observe events in $E_o\setminus E_I$ and $E_o\setminus E_D$;
%     \item Pruning away states that violate the properties of $i$-availability and confidentiality in the edit game structure and resulting in a trimmed game structure;
%     \item Merging states that are reached by sequences unobserved by the edit function in the trimmed game structure such that a no-guarantees edit mechanism is constructed;
%     \item Pruning away states that violate the property of $c$-availability in the no-guarantees edit mechanism and resulting in an edit mechanism.
% \end{enumerate}
% In the subsequent Sections~\ref{sec:game} and \ref{sec:mechanism}, the structure of the corresponding stage will be defined and the algorithms will be provided.

The construction of an edit mechanism involves a systematic process, depicted in Fig. \ref{Constructionflow}, comprising the following key steps:

\begin{figure}[htbp]
    \centering
    \includegraphics[width=0.6\linewidth]{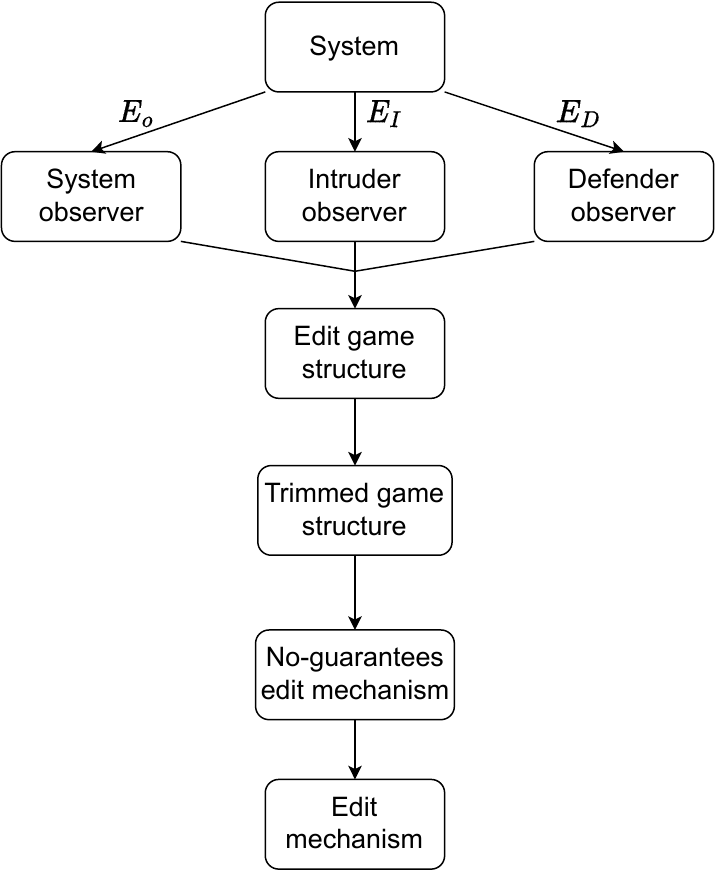}
    \caption{The construction flow of an edit mechanism.}
    \label{Constructionflow}
\end{figure}

\begin{enumerate}
    \item \textbf{Constructing Observers:} Develop observers (\emph{system, intruder, and defender observers}) to estimate the current state of the given system from the perspectives of the system, the intruder, and the defender.
    \item \textbf{Game Structure Construction:} Integrate the system, intruder, and defender observers in a game-like fashion, creating a so-called \emph{edit game structure}. This involves assuming that the intruder and the defender can observe the events in \(E_o\setminus E_I\) and \(E_o\setminus E_D\), respectively.
    \item \textbf{Pruning States:} Trim the edit game structure by removing the states that violate the properties of \(i\)-availability and confidentiality, resulting in a more refined structure, called \emph{trimmed game structure}.
    \item \textbf{Merging States:} Combine the states reached by the sequences unobserved by the edit function in the trimmed game structure. This step leads to the creation of a \emph{no-guarantees edit mechanism}.
    \item \textbf{Final Refinement:} Refine the no-guarantees edit mechanism by pruning away the states that violate the property of \(c\)-availability such that an \emph{edit mechanism} is obtained.
\end{enumerate}
In the subsequent Sections~\ref{sec:game} and \ref{sec:mechanism}, the structure of each step will be precisely defined, accompanied by detailed algorithms to facilitate the implementation of this construction process.

\section{Synthesis of an edit game structure}
\label{sec:game}

This section presents the procedure of building an edit game structure that enumerates all edit actions associated with the system behavior. 
To this end, we first introduce intruder and defender observers to characterize the estimates of a system from the perspectives of the intruder and an edit function. 
Then, the edit game structure is built by considering incomparable observations between the intruder and an edit function. 
In the end, we propose a utility function to identify problematic states reached by undesirable edit actions in the game structure.

\subsection{System Estimate by the Intruder and the Edit Function}

Given a system $G=(X,E,f,x_0)$ and a sequence $\sigma\in E_o^*$, the set of possible states with respect to $P$ starting from a state $x\in X$ is defined as $R_o(x,\sigma)=\lbrace x'\in X\mid \exists \lambda\in E^*:P(\lambda)=\sigma, f(x,\lambda)=x'\rbrace$. 
One can build a system observer (in the standard manner provided in \cite{2020Chrisbook}) to estimate the current states of the system under $E_o$, which is defined as $\mathbf{O}_o=(X_o,E_o,f_o,x_o^0)$, where $X_o\subseteq 2^X$ is the state space, $E_o$ is the set of observable events, $x_o^0=R_o(x_0,\varepsilon)$ is the initial state, and $f_o: X_o\times E_o\rightarrow X_o$ is the transition function defined for $x_o\in X_o$ and $\alpha\in E_o$ as $f_o(x_o,\alpha)=\bigcup_{x\in x_o}R_o(x,\alpha)$ (by convention $f_o(x_o,\alpha)$ is taken to be undefined if $R_o(x_o,\alpha)$ is empty).

%Intruder and defender observers can be constructed in a manner similar to the standard observer proposed in \cite{2010Lafortunebook,2020Chrisbook}.
Given a system $G=(X,E,f,x_0)$ and a sequence $\sigma\in E_I^*$, the set of possible states with respect to $P_I$ starting from a state $x\in X$ is defined as $R_I(x,\sigma)=\lbrace x'\in X\mid \exists \lambda\in E^*:P_I(\lambda)=\sigma, f(x,\lambda)=x'\rbrace$. 
One can build an intruder observer $\mathbf{O}_I$ to estimate the current states of the system under $E_I$ as follows (we allow the events in $E_o\setminus E_I$ to perform a self-loop in $\mathbf{O}_I$ for notational convenience).

\begin{definition}
    Given a system $G=(X,E,f,x_0)$ with respect to $E_o$ and $P_I$, an intruder observer is defined as $\mathbf{O}_I=(X_I,E_o,f_I,x_I^0)$, where $X_I\subseteq 2^X$ is the state space, $E_o$ is the set of observable events, $x_I^0=R_I(x_0,\varepsilon)$ is the set of initial states, and $f_I: X_I\times E_I\rightarrow X_I$ is the transition function defined for $x_I\in X_I$ and $\gamma\in E_I$ as $f_I(x_I,\gamma)=\bigcup_{x\in x_I}R_I(x,\gamma)$, and $f_I(x_I,\beta)=x_I$ for $\beta\in E_o\setminus E_I$.
\end{definition}

Similarly, based on the system $G=(X,E,f,x_0)$ and a sequence $\sigma\in E_D^*$, the set of possible states with respect to $P_D$ starting from a state $x\in X$  is defined as $R_D(x,\sigma)=\lbrace x'\in X\mid \exists \lambda\in E^*:P_D(\lambda)=\sigma, f(x,\lambda)=x'\rbrace$. 
Then, a defender observer, denoted by $\mathbf{O}_D$, is given by the following definition (we allow the events in $E_o\setminus E_D$ to perform a self-loop in $\mathbf{O}_D$ for notational convenience).

\begin{definition}\label{DefOD}
	Given a system $G=(X,E,f,x_0)$ with respect to $E_o$ and $P_D$, and an edit function $f_e$, a defender observer is defined as $\mathbf{O}_D=(X_D,E_o,f_D,x_D^0)$, where  $X_D\subseteq 2^X$ is the state space, $E_o$ is the event set, $x_D^0=R_D(x_0,\varepsilon)$ is the set of initial states, and $f_D: X_D\times E_o\rightarrow X_D$ is the transition function defined as $f_D(x_D,f_e(\gamma))=\bigcup_{x\in x_D}R_D(x,f_e(\gamma))$ for $x_D\in X_D$ and $\gamma\in E_D$ if $f_e(\gamma)$ is defined, and $f_D(x_D,\gamma)=x_D$ for $\gamma\in E_o\setminus E_D$.
\end{definition}

Recall that an edit function confuses the intruder by creating a perturbed output sequence. 
This is captured by the transition function $f_D$ that takes as input the edited outputs.
Since the edit function can only manipulate the events in $E_D$, the transition function $f_D(x_D,f_e(\gamma))$ implements the edit actions if i) observation $\gamma\in E_D$ for $x_D\in X_D$ is received, and ii) $f_D(x_D,f_e(\gamma))$ is defined.
Thus, the transition function $f_D(x_D,f_e(\gamma))$ ensures that the defender observer contains all edit actions in response to every event observed by an edit function; a self-loop is added for all events unobserved by the edit function at each state.

%defended observer that can be constructed by replacing all events in $\mathbf{O}_D$ that are observed by the edit function with the corresponding edited operations. 

% \begin{definition}\label{DefD}
% 	Given a system $G=(X,E,f,x_0)$ with respect to $E_D$, $P_D$ and $f_e$, a defended observer is constructed by $\mathbf{O}_T=(X_T,E_{RH},f_T,x_T^0)$, where $X_T=X_D$ is the state space, $x_T^0=x_D^0$ is the set of initial states, $E_{RH}$ is the set of events, and $f_T$ is the transition function that implements edited operations as $f_T(x_T,f_e(\gamma))=f_D(x_T,f_e(\gamma))$ for $x_T\in X_T$ and $\gamma\in E_D$ if $f_e(\gamma)$ is defined; and $f_T(x_T,\beta)=x_T$ for $\beta\in E_{RH}\setminus E_D$.
% \end{definition}

\begin{eg}
	Consider the system $G$ as shown in Fig. \ref{system}, where the set of secret states is $X_S=\lbrace 5\rbrace$ and the set of observable events is $E_o=\lbrace a,b,c,d\rbrace$.
    In this regard, the system observer is identical to the system itself, i.e., $\mathbf{O}_o=G$.

    Assume that the set of events observed by the intruder is $E_I=\lbrace a,b,d\rbrace\subseteq E_o$.
    One can construct the intruder observer $\mathbf{O}_I$ as shown in Fig. \ref{DiagramOI}; all the events unobserved by the intruder are added as a self-loop at each state, e.g., event $c$ at the initial state.
    Following Theorem~1 in \cite{2007Chris}, we conclude that CSO is violated since there exists a solely secret state $\lbrace5\rbrace$ in~$\mathbf{O}_I$.

    Assume that the set of events observed by an edit function is $E_D=\lbrace b,c,d\rbrace\subseteq E_o$.
    One can construct the defender observer $\mathbf{O}_D$ as shown in Fig. \ref{DiagramOD}.
    From the initial state $\lbrace 1,3\rbrace$ in $\mathbf{O}_D$, the new state is state $\lbrace 2,5\rbrace$ if the edit function receives observation $b$ and replaces $b$ with itself via $f_e(b)=b$ or if it receives observation $c$ and replaces $c$ with $b$ via $f_e(c)=b$.
	There is no update if the edit function inserts $d$ before $b$ via $f_e(b)=db$ since $db$ is not defined (i.e., we cannot find sequence $db$ from the initial state in $\mathbf{O}_D$); or if it receives an observation that cannot be observed, i.e., event~$a$.
 $\hfill\diamond$
\end{eg}

    By taking advantage of the constructions of $\mathbf{O}_I$ and $\mathbf{O}_D$, if the system generates an event unobserved by the intruder (the edit function), the intruder is assumed to ``observe'' (the edit function is assumed to ``react'', i.e., output the same event) via self-loops in the corresponding states.
    Also, one can notice that the estimate of the intruder is not necessarily the same as the estimate of the edit function since their observations are incomparable.

%  \graphicspath{{./figure/}}
% \begin{figure}[htbp]
% 	\centering
% 	\includegraphics[width=0.2\textwidth]{System.pdf}
% 	\caption{System $\emph G$ used in the running example.}
% 	\label{system}
% \end{figure}

\subsection{Edit Game Structure}

To systematically illustrate how an edit function can execute edit actions following system behavior, we first assume that the intruder observes the events in $E_o\setminus E_I$ and the edit function observes (but cannot react to) the events in $E_o\setminus E_D$, which has been captured via the self-loops in the constructions of $\mathbf{O}_I$ and $\mathbf{O}_D$.
Then, a two-player game structure between the system and the edit function is constructed as follows.

\begin{definition}
	\label{editGS}
	Consider a system $G=(X,E,f,x_o)$ with a set of secret states $X_S$, system observer $\mathbf{O}_o=(X_o, \allowbreak   E_o,\allowbreak f_o,\allowbreak x_o^0)$, 
    intruder observer $\mathbf{O}_I=(X_I,\allowbreak E_o,\allowbreak f_I,\allowbreak x_I^0)$, and defender observer $\mathbf{O}_D=(X_D,\allowbreak E_o,\allowbreak f_D, \allowbreak x_D^0)$. 
    An edit game structure is defined as $\mathcal{EGS} = (V, \allowbreak E_o\cup E_o^{\leq (K+1)},\allowbreak  \delta_{ID}\cup\delta_{DI},\allowbreak  v_0)$, where
	\begin{enumerate}
		\item  $V=V_A\cup V_F$ with $V_A= X_o\times X_I\times X_D$ being the set of information states and $V_F=(X_o\times X_I\times X_D)\times E_o$ being the set of information states augmented with observable events;
		\item   $v_0=(x_o^0, x_I^0, x_D^0)\in V_A$ is the initial state;
		\item   $E_o$ is the set of actions for the system and $E_o^{\leq (K+1)}$ is the set of actions for the edit function; 
		\item  $\delta_{ID}=V_A\times E_o \rightarrow V_F$ is the transition function from the system to the edit function defined as: for all $(x_o, x_I, x_D)\in V_A$, for all $\alpha\in E_o, \delta_{ID}((x_o, x_I, x_D),\alpha) =  ((f_o(x_o,\alpha),x_I,x_D),\alpha)$ if $f_o(x_o, \alpha)$ is defined; 
        \item  $\delta_{DI}: V_F\times E_o^{\leq (K+1)}\rightarrow V_A$ is the transition function from the edit function to the system defined as: \\
        a) for all $((x_o, x_I, x_D), \alpha)\in V_F$ with $\alpha\in E_o\setminus (E_I\cup \allowbreak E_D)$, $\delta_{DI}(((x_o, x_I, x_D), \alpha), \alpha)=(x_o, x_I, x_D)$;\\
        b) for all $((x_o, x_I, x_D), \alpha)\in V_F$ with $\alpha\in E_I\setminus E_D$, $\delta_{DI}(((x_o, x_I, x_D), \alpha), \alpha)=(x_o, f_I(x_I,\alpha), x_D)$ if $f_I(x_I, \alpha)$ is defined;\\
        c) for all $((x_o, x_I, x_D), \alpha)\in V_F$ with $\alpha\in E_D\setminus E_I$, for all $\omega\in E_D^{\leq (K+1)}$, we have i) $\omega=\alpha'\in E_D\setminus\lbrace\alpha\rbrace$ via $f_e(\alpha)=\omega$ (substitution), ii) $\omega=\varepsilon$ via $f_e(\alpha)= \varepsilon$ (deletion), and iii) $\omega= \omega'\alpha$ via $f_e(\alpha)= \omega'\alpha$ (insertion), such that $\delta_{DI}(((x_o,x_I,x_D), \alpha), \omega)=((x_o, f_I(x_I,\omega), f_D(x_D, \omega))$ for $\omega\in (E_D\cap E_I)^*$ if both $f_I(x_I, \omega)$ and $f_D(x_D, \omega)$ are defined and $\delta_{DI}(((x_o,x_I,x_D), \alpha), \omega)=((x_o, x_I, f_D(x_D, \omega))$ for $\omega\in (E_D\setminus E_I)^*$ if $f_D(x_D, \omega)$ is defined;\\
        d) for all $((x_o, x_I, x_D), \alpha)\in V_F$ with $\alpha\in E_D\cap E_I$, for all $\omega\in E_D^{\leq (K+1)}$, we have i) $\omega=\alpha'\in E_D\setminus\lbrace\alpha\rbrace$ via $f_e(\alpha)=\omega$ (substitution), ii) $\omega=\varepsilon$ via $f_e(\alpha)= \varepsilon$ (deletion), and iii) $\omega= \omega'\alpha$ via $f_e(\alpha)= \omega'\alpha$ (insertion), such that $\delta_{DI}(((x_o,x_I,x_D), \alpha), \omega)=((x_o, f_I(x_I,\omega), f_D(x_D, \omega))$ for $\omega\in (E_D\cap E_I)^*$ if both $f_I(x_I, \omega)$ and $f_D(x_D, \omega)$ are defined and $\delta_{DI}(((x_o,x_I,x_D), \alpha), \omega)=((x_o, x_I, f_D(x_D, \omega))$ for $\omega\in (E_D\setminus E_I)^*$ if $f_D(x_D, \omega)$ is defined.\\
	\end{enumerate}
\end{definition}

Recall that $E_o^{\leq (K+1)}\subset E_o^*$ denotes the finite set of strings of maximum length $K+1$, where $K\in\mathbb{N}$ is the maximum number for insertions.
In this regard, $E_o^{\leq (K+1)}$ can represent all edit actions implemented by edited operations (including substitutions, deletions, and finite insertions) and events that the edit function cannot observe during the evolution of the transition function $\delta_{ID}: V_F\times E_o^{\leq (K+1)}\rightarrow~V_A$.

For any state $(x_o,x_I,x_D)\in V_A$, we say that it is an information state that contains the state estimates of the system, the intruder, and the edit function.
Then, the system can play by executing an observable event (e.g., $\alpha\in E_o$) such that state $((x_o',x_I,x_D),\alpha)\in V_F$ is reached via $\delta_{ID}((x_o,x_I,x_D),\alpha)=((x_o',x_I,x_D),\alpha)$.
At this point, only state $x_o'$ is updated from $x_o$ via $f_o(x_o,\alpha)=x_o'$ and event $\alpha$ is reserved to help determine what edit actions can be implemented by the edit function in the next step.

\graphicspath{{./figure/}}
\begin{figure*}[htbp]
	\centering
	{
		\includegraphics[width=0.9\textwidth]{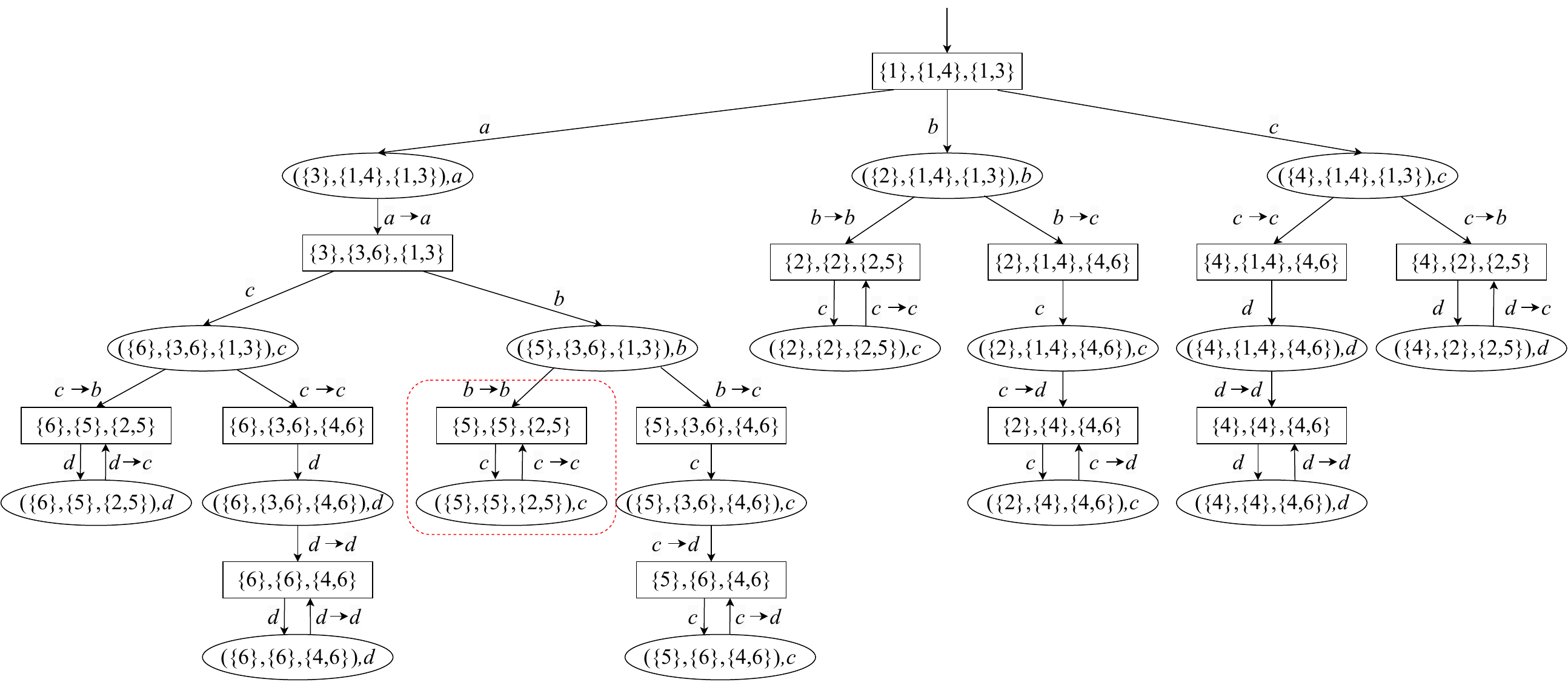}
	}	
	\caption{Edit game structure.}
 \label{EGS}
\end{figure*}

Then, the edit function can play by executing an edit action (i.e., $\omega\in E_o^{\leq (K+1)}$) such that a new information state $(x_o',x_I',x_D')\in V_A$ can be reached via $\delta_{DI}(((x_o',x_I,x_D),\alpha),\omega)=(x_o',x_I',x_D')$, where $x_I'$ and $x_D'$ may be updated depending on how the edit function is implemented. 
Specifically, there are four cases\footnote{
Since edit action $\omega$ generated by the edit function could be a string (due to insertion operations), we also have the case where some elements of $\omega$ are in $(E_o\cap E_I)^*$ and others are in $(E_D/E_I)^*$.
While this introduces additional complexity, it can be captured by the transition function $\delta_{DI}$ in Definition 10 in a straightforward manner.
For simplicity, we only consider four cases as depicted in the definition of $\delta_{DI}$.} that have to be taken into account due to the incomparable observation setting: 
\begin{enumerate}
	\item When the system executes an event $\alpha\in E_o\setminus (E_I\cup E_D)$, this means that it cannot be observed by the intruder or the edit function. 
    Thus, the edit function can only output $\alpha$ such that the estimates of the intruder and the edit function remain unchanged.
	\item When the system executes an event $\alpha\in E_I\setminus E_D$, this means that it can only be observed by the intruder.
    Thus, the edit function can only ``output'' $\alpha$ such that the estimate of the intruder updates and the estimate of the edit function remains unchanged.
	\item When the system executes an event $\alpha\in E_D\setminus E_I$, this means that it can only be observed by the edit function.
    Thus, the edit function can implement edited operations to modify $\alpha$ to $\omega$. 
    If $\omega$ belongs to $(E_D\cap E_I)^*$, the intruder can observe the modified observations such that the estimates of the intruder and the edit function update; if $\omega$ belongs to $(E_D\setminus E_I)^*$, the intruder cannot observe the modified observations such that the estimate of the edit function updates and the estimate of the intruder remains unchanged. 
	\item When the system executes an event $\alpha\in E_D\cap E_I$, both the intruder and the edit function can observe it. 
    Thus, the edit function can implement edited operations to modify $\alpha$ to $\omega$. 
    If $\omega$ belongs to $(E_D\cap E_I)^*$, the intruder can observe the modified observations such that the estimates of the intruder and the edit function update; if $\omega$ belongs to $(E_D\setminus E_I)^*$, the intruder cannot observe the modified observations such that the estimate of the edit function updates and the estimate of the intruder remains unchanged.
\end{enumerate}

In this regard, given the state $(x_o',x_I',x_D')\in V_A$, state $x_o'$ represents the actual estimate of the system, whereas states $x_I'$ and $x_D'$ represent, respectively, the fake estimate of the system from the viewpoints of the intruder and the edit function since the updating of their states is based on the edit action.
With alternate moves between the system and the edit function, $\mathcal{EGS}$ is constructed such that it contains all possible edit actions following system behavior. 

%Notice that the insertion and deletion operations implemented by the edit function can be represented as $\omega=\omega'\alpha\in E_D^*E_D$ and $\omega=\varepsilon\in E_D^*$.
%For uniformity, we use $\omega\in E_D^*$ to denote any edited operation reacted by the edit function.

\begin{eg}\label{egEGS}
	Consider again the system $G$ in Fig. \ref{system} and its system observer, intruder observer, and defender observer in Fig. \ref{DiagramOIOD}.
	Following Definition \ref{editGS}, one can construct the edit game structure $\mathcal{EGS}$ by taking into account the incomparable observations between the intruder and the edit function. 
	To save space, we only consider the substitution operations for the edit function during the construction of $\mathcal{EGS}$ as shown in Fig. \ref{EGS}.
    We use the arrows on the edited operations to represent how the edit function works, e.g., $b\rightarrow c$ means replacing $b$ with $c$.

	Specifically, from the initial state $(\lbrace1\rbrace,\lbrace1,4\rbrace,\lbrace1,3\rbrace)$, the $\mathcal{EGS}$ evolves by following the system behavior (i.e., events $a$, $b$, or $c$).
	If event $a$ occurs, then the initial state will update to $((\lbrace 3\rbrace,\lbrace 1,4 \rbrace,\lbrace1,3\rbrace),a)$ via  $f_o(\lbrace 1\rbrace,a)=\lbrace3\rbrace$.
    At this point, only the estimate of the system observer is updated.
    Then, the edit function can only output $a$ without any edited operation (since it cannot observe $a$).
    Thus, $\mathcal{EGS}$ moves to state $(\lbrace3\rbrace,\lbrace3,6\rbrace,\lbrace1,3\rbrace)$, where the estimate of the intruder observer is updated since the intruder can observe~$a$.
    With alternate moves between the system and the edit function, the rest of the construction $\mathcal{EGS}$ can be completed.

    As an example of how the effects of incomparable observations between the intruder and the edit function are reflected in the construction  $\mathcal{EGS}$, let us consider state $(\lbrace3\rbrace,\lbrace3,6\rbrace,\lbrace1,3\rbrace)$: two events can be executed by the system, i.e., events $b$ and $c$, where $b$ is observed by both the edit function and the intruder while $c$ is only observed by the edit function.
    If event $b$ is executed, the edit function could output $c$ by replacing $b$ with $c$ (which the intruder cannot observe), resulting in a new state $(\lbrace5\rbrace,\lbrace3,6\rbrace,\lbrace4,6\rbrace)$.
    On the other hand, if action $c$ is executed, the edit function could output $b$ by replacing $c$ with $b$ (which the intruder can observe), resulting in a new state $(\lbrace6\rbrace,\lbrace2\rbrace,\lbrace2,5\rbrace)$.
    Thus, a more sophisticated design is needed for the edit function under the incomparable observations to handle all above cases. 
	$\hfill\diamond$
\end{eg}

\subsection{Utility Function}

Although an edit game structure $\mathcal{EGS}$ contains all possible edit actions following the system behavior, not all of them are feasible.
It is important to note that $\mathcal{EGS}$ may give rise to problematic states under two specific cases:
\begin{enumerate}
    \item When an edit function may be incapable of responding to one or more events generated by a system.
    \item When the intruder can deduce the secrets after edit actions performed via an edit function.
\end{enumerate}

In the former case, the $i$-availability of an edit function is compromised, while in the latter, the confidentiality of an edit function is violated.
To identify these problematic states based on the edit game structure, we introduce a utility function as a state-based binary function. 
This process involves an evaluation for each information state, primarily to check if the secret state is revealed from the intruder's estimation. 
Additionally, we check each information state with augmented actions to identify the cases where no edit action can respond to the given system action.

% When the edit function is aware of the existence of the intruder, we are able to identify unfeasible edited operations from the edit game structure when the edit function is unknown.
% That is, we check the intruder's estimate on each information state if the secret state solely exists, and we check each information state with augmented actions if there is no edited operation to respond.

\begin{definition}\label{UF}
	Consider an edit game structure $\mathcal{EGS} = (V,  E_o\cup E_o^{\leq (K+1)},  \delta_{ID}\cup\delta_{DI}, v_0)$, where $V=V_A\cup V_F$. Suppose that an edit function is unknown. A utility function, denoted by $U:V_A\cup V_F\rightarrow\lbrace 0,1\rbrace$, is defined for state $v\in V_A\cup V_F$ such that:
\begin{equation*}
    U(v) = \left\{
    \begin{aligned}
        0, &\quad \text{if } [[v=(x_o,x_I,x_D)\in V_A] \wedge [x_o\subseteq X_S] \wedge \\
           & [x_I\subseteq X_S]] \vee [[v=((x_o,x_I,x_D),\alpha)\in V_F] \wedge \\
           & [(\forall \omega\in E_o^{\leq (K+1)}) [\delta_{DI}(v,\omega)=\emptyset]]], \\
        1, &\quad \text{otherwise.}
    \end{aligned}
    \right.
\end{equation*}
\end{definition}

\vspace{0.5em}

The utility function $U$ captures two problematic states from each of the two types of states in the edit game structure. 
Specifically, state $v=(x_o,x_I,x_D)\in V_A$ is assigned value~0 if the intruder's estimate is a subset of secret states while the system's estimate is also a subset of secret states, i.e., $[x_o\subseteq X_S]\wedge[x_I\subseteq X_S]$.
This is due to the fact that the state $(x_o, x_I, x_D)$ compromises the confidentiality of the edit function, as it allows the intruder to deduce that the system is in secret states while the system is indeed in secret states.
A state of the form $v=[(x_o, x_I, x_D),\alpha]\in V_F$ is assigned value~0 if for all edit actions in the form $\omega\in E_o^{\leq (K+1)}$, $\delta_{DI}(v,\omega)$ is not defined; this means that the edit function is not able to react to event~$\alpha$.
The availability of the edit function is violated at this point.

\begin{eg}
    Continuing with Example \ref{egEGS}, we can compute the problematic states in the edit game structure $\mathcal{EGS}$ via the utility function $U$.
    For instance, state $(\lbrace 5\rbrace,\lbrace 5\rbrace,\lbrace 2,5\rbrace)$ in Fig.~\ref{EGS} is problematic since the intruder observer is at state $\lbrace 5\rbrace$, which means that the intruder can infer the secret state.
    Note that, in this example, there is no problematic state that violates $i$-availability since every event generated by the system has at least one edit action.
        \hfill
        $\diamond$
\end{eg}

\begin{remark}
    In \cite{2014Wu, 2018Wu, 2018Ji, 2019Ji, 2023Duan}, the edit function is required to output a sequence that reaches non-secret states so as to ensure that the intruder cannot deduce the secret states.
    In other words, not only state $(\lbrace 5\rbrace,\lbrace 5\rbrace,\lbrace 2,5\rbrace)$ but state $(\lbrace 6\rbrace,\lbrace 5\rbrace,\lbrace 2,5\rbrace)$ in Fig. \ref{EGS} should be treated as problematic since the system is in secret state $\lbrace5\rbrace$ from the point of view of the intruder.
    In this work, we consider that the objective of the edit function is to confuse the intruder.
    That is, state $(\lbrace 6\rbrace,\lbrace 5\rbrace,\lbrace 2,5\rbrace)$ is admissible since the system is not in a secret state in reality; only state $(\lbrace 5\rbrace,\lbrace 5\rbrace,\lbrace 2,5\rbrace)$ is treated as problematic (because the intruder thinks the system is in a secret state while the system is indeed in a secret state).
    In this regard, we say that the implementation of the edit function under the objective considered in this work is more flexible compared to the prior works \cite{2014Wu, 2018Wu, 2018Ji, 2019Ji, 2023Duan}.
    However, one can easily adjust our approach to achieve the objectives of this earlier work.
    \hfill
$\diamond$
\end{remark}

\section{Synthesis of edit mechanism}
\label{sec:mechanism}

In this section, our aim is to build an edit mechanism derived from an edit game structure such that it contains all feasible edit actions following system behavior observed by an edit function.
To do so, we first construct a trimmed-version of the edit game structure, named a trimmed game structure, by pruning away all problematic states (identified via the utility function) such that all feasible edit actions can be obtained.
Then, the edit mechanism is constructed based on the edit function's observations. 
This process involves a two-step procedure.
At the first step, in the trimmed game structure, we merge the states that are reached by sequences unobserved by the edit function, obtaining a structure, called a no-guarantees edit mechanism, where all feasible edit actions can be implemented by the edit function to react to system behavior that can be observed by the edit function.
At the second step, in the no-guarantees edit mechanism, we check if all edit actions can be recognized by the intruder such that the intruder cannot infer the secret states under the incomparable observation setting (also refer to Fig. \ref{Constructionflow}).

\subsection{Trimmed Game Structure}

The utility function defines two types of problematic states, and problematic states that violate the $i$-availability condition should be pruned away in a recursive manner.
This process can be interpreted as an instance of a  \emph{supervisory control problem without blocking} (SCPB); more details can be found in \cite{cassandras2021introduction}. 

To guarantee that the correct decision can be taken at each state, all states violating the utility function (i.e., with values equal to zero) should fall into the set of problematic states. 
As we are unable to disable the events generated by the system that lead to problematic states, our sole vehicle is to disable edit actions, i.e., the events generated by the edit function, that result in antecedents to these problematic states.
%Since we cannot disable the events generated by the system that lead to the problematic states, the only thing we can do is to disable edit actions (i.e., events generated by the edit function) that lead to antecedents of these problematic states. 
When disabling events, however, new problematic states can be obtained.
We have to iterate the process until no additional such states can be generated. 
In this regard, SCPB can be implemented by a control policy, which is defined as $C:V_F\rightarrow 2^{E_o^{\leq (K+1)}}$, based on information states augmented with observable events in the edit game structure.
For a state $v_F\in V_F$, $C(v_F)$ is the set of controllable events enabled at $v_F$.
Algorithm 1 formally presents the construction process of the
trimmed game structure $\mathcal{TGS}$.

\begin{algorithm}[!htb]	
	\caption{Construction of a trimmed game structure $\mathcal{TGS}$}
	\label{alg_TEM}
	\LinesNumbered
	\KwIn{$\mathcal{EGS}=(V, E_o\cup E_o^{\leq (K+1)}, \delta, v_0)$ with $V=V_A\cup V_F$, $\delta=\delta_{ID}\cup\delta_{DI}$.}
	\KwOut{$\mathcal{TGS}= (V_T, E_o, \delta, v_0)$.}
	Initialize $V_{PA}=\lbrace v_A\in V_A| U(v_A)=0\rbrace$ and $V_{PF}=\lbrace v_F\in V_F| U(v_F)=0\rbrace$\;
    \tcp{$V_{PA}$ ($V_{PF}$) is the set of problematic states in $V_A$ ($V_F$)}
	\For{$v_A\in (V_A\setminus V_{PA})$}{
		\If{$\exists \alpha\in E_o$, $\delta_{ID}(v_A,\alpha)\in V_{PF}$}{
			$V_{PA}=V_{PA}\cup\lbrace v_A\rbrace$\;
		}
	}
	\For{$v_F\in (V_F\setminus V_{PF})$}{
		\If{$\forall \omega\in E_o^{\leq (K+1)}$, $\delta_{DI}(v_F,\omega)\in V_{PA}$}{
			$V_{PF}=V_{PF}\cup\lbrace v_F\rbrace$\;
		}
	}
	Return to line 2 and repeat until no new problematic state is produced\; 
	\For{$v_F\in (V_F\setminus V_{PF})$}{
		\For{$\omega\in E_o^{\leq (K+1)}$}{
			\If{$\delta_{DI}(v_F,\omega)\in V_{PA}$}{
				$C(v_F)=C(v_F)\setminus \omega$\;
			}
		}
	}
   Return $\mathcal{TGS}= (V_T, E_o\cup E_o^{\leq (K+1)}, \delta, v_0^T)$.
   \tcp{$V_T=V\setminus (V_{PA}\cup V_{PF})$ and $v_0^T=v_0$}
\end{algorithm}

\begin{lemma}\label{lemma1}
    An edit function is $i$-available and confidential if and only if it is synthesized from a non-empty trimmed game structure $\mathcal{TGS}$.
\end{lemma}

\begin{proof}
	(If) By contradiction, we assume that the edit function is not $i$-available or confidential.
	Thus, it cannot react to every event observed by the edit function or it cannot modify the sequences generated by the system to edited sequences such that the intruder cannot infer the secret. 
	In other words, one can find a sequence leading to a state $v$ whose the utility function satisfies $U(v)=0$.
	However, by construction of $\mathcal{TGS}$, such sequences have been removed, which is a contradiction.
	Therefore, the edit function should be $i$-available and confidential.
 
	(Only if) Given an edit function that is $i$-available and confidential, it can be retained in $\mathcal{EGS}$ since we build $\mathcal{EGS}$ by following system behavior observed by the edit function. 
	Moreover, we prune away all the states violating the utility function $U$ when we build $\mathcal{TGS}$.
	Therefore, one can conclude that the $i$-available and confidential edit function can be synthesized from $\mathcal{TGS}$.
\end{proof}

\begin{eg}
	Continuing with Example 3, the trimmed game structure $\mathcal{TGS}$ can be constructed from $\mathcal{EGS}$ via SCPB by following Algorithm 1.
    First, the problematic state $(\lbrace5\rbrace,\lbrace5\rbrace,\lbrace2,5\rbrace)$ is obtained via the utility function $U$ in Example 4. 
    Then, edit action $b\rightarrow b$ is disabled at state $[(\lbrace5\rbrace,\lbrace3,6\rbrace,\lbrace1,3\rbrace),b]$ to prevent this problematic state since $b\rightarrow b$ is controllable (since $b$ can be observed by the edit function).
    Note that we do not need to remove any other state since there is no problematic state that violates $i$-availability in this example (in principle, we need to remove such states in a recursive manner if new problematic states are generated after pruning).

    In conclusion, one can prune away all problematic states by following the utility function, which results in the specification for the supervisory control problem.
	The resulting trimmed game structure $\mathcal{TGS}$ is shown in Fig. \ref{EGS} after removing the dotted box with red color.
	$\hfill\diamond$
\end{eg}

\subsection{Edit Mechanism}

Recall that an edit game structure is constructed under the assumption that both the intruder and an edit function possess full observation capability of the system. 
This assumption allows us to create the trimmed game structure, which ensures that every observable sequence generated by the system can be responded by at least one edit action while preventing the system's secrets from exposure.
However, due to the presence of the incomparable observation setting,  we then have to remove the assumption that the edit function has full observation capability of the system, and, instead, focus on its own observations (partial observation capability of the system).

To this end, we introduce an edit mechanism, which is constructed by merging the states from $\mathcal{TGS}$, so as to: (i) incorporate the fact that the events in $E_o\setminus E_D$ are not observable to the edit function;
(ii) ensure that all edit actions are defined at the merged states (otherwise the intruder will infer the existence of the edit function when/if the edit function outputs an edit action that is not recognized from the point view of the intruder).
In other words, we must account for situations where the edit function employs the same edit action as derived from its own observations.
%In essence, this edit mechanism addresses situations where the edit function employs the same edited operations as those derived from its own observations.
The formal procedures are presented in Algorithms \ref{alg_EMs1} and \ref{alg_EMs2}.
%(ii) ensure that all edited operations cannot lead to inconsistency that will be recognized by the intruder.

\begin{algorithm}[!htb]	
	\caption{Construction of a no-guarantees edit mechanism $\mathcal{UEM}$}
	\label{alg_EMs1}
	\LinesNumbered
	\KwIn{$\mathcal{TGS}=(V_T, E_o\cup E_o^{\leq (K+1)}, \delta, v_0^T)$, where $V_T=V_{TA}\cup\allowbreak V_{TF}\subseteq V_A\cup V_F$ and $\delta=\delta_{ID}\cup\delta_{DI}$.}
	\KwOut{$\mathcal{UEM}=(V_U,E_D\cup E_D^{\leq (K+1)},\delta_I\cup\delta_D, v^U_0)$ with $V_U=V_{UA}\cup V_{UF}$.}
	Compute $v^U_0=\{v_0^T\}\cup \{v_{TA}\in V_{TA}\mid\exists \sigma\in(E_o\backslash E_D)^*:\delta(v_0^T,\sigma)=v_{TA}\}$\;
	Initialize $V_U=V_{UA}=\lbrace v^T_0\rbrace$, and $V_{UF}=\emptyset$\;
	\For{$v_{UA}\in V_{UA}$ that have not been examined}{
		\For{$\gamma\in E_D$}{
				$\delta_I(v_{UA},\gamma)=\bigcup_{z\in v_{UA}}\delta_I(z,\gamma)=\lbrace y\in V_{TF}\mid \exists \sigma\in E_o^*: P_D(\sigma)=\gamma\wedge y\in \delta(z,\sigma)\rbrace$ \;
				Add $\delta_I(v_{UA},\gamma)$ to $V_{UF}$\;
			}
		}
	\For{$v_{UF}\in V_{UF}$ that have not been examined}{
		\For{$\omega\in E_D^{\leq (K+1)}$}{
				$\delta_D(v_{UF},\omega)=\bigcup_{z\in v_{UF}}\delta_D(z,\omega)=\bigcup_{z\in v_{UF}}\delta_{DI}(z,\omega)$\;
				Add $\delta_D(v,\omega)$ to $V_{UA}$\;
			}    
		}    
	Go back to line 3 and repeat until all accessible part has been built\;
\end{algorithm}

Algorithm \ref{alg_EMs1} presents the procedure of merging states from $\mathcal{TGS}$ to obtain all edit actions following the system behavior that can be observed by the edit function.
Specifically, we first compute the initial state of the edit mechanism $v_0^U$ in line 1 as the set of states that can be reached from state $v_0^T$ in $\mathcal{TGS}$ via  sequences of events that are unobserved by the edit function.
Then, lines 3 to 6 evolve the initial state $v_0^T$ to a new state $v_1^{TA}\in V_{UA}$ via $\delta_I$ when the edit function receives observation $\gamma\in E_D$ from the sequences $\sigma=\sigma_1\gamma\sigma_2$ generated by the system in $\mathcal{TGS}$, where $\sigma_1,\sigma_2\in (E_o\setminus E_D)^*$.
Lines 7 to 10 evolve state $v_1^{UA}$ to a new state $v_2^{UF}\in V_{UF}$ via $\delta_D$ when the edit function manipulates observation $\gamma$ to $\omega\in E_D^{\leq (K+1)}$ in terms of $f_e(\gamma)=\omega$.
In this regard, the no-guarantees edit mechanism contains all edit actions following system behavior that can be observed by the edit function.

Note that, under the incomparable observation setting, some edit actions might not be recognized from the point view of the intruder, which may allow the intruder to infer the presence of the edit function.
To mitigate this, we propose Algorithm~\ref{alg_EMs2} to check if any edit action can be recognized by the intruder.
If so, we remove such edit actions from the no-guarantees edit mechanism and obtain a new edit mechanism.
Note that this process can also be regarded as an instance of SCPB since it should be executed in a recursive manner.
To do so, we redefine the domain of the control policy as $C:V_{UF}\rightarrow 2^{E_D^{\leq (K+1)}}$, where $C(v_{UF})$ is the set of controllable events enabled at $v_{UF}\in V_{UF}$.

\begin{algorithm}[!htb]	
	\caption{Construction of an edit mechanism $\mathcal{EM}$}
	\label{alg_EMs2}
	\LinesNumbered
	\KwIn{$\mathcal{UEM}=(V_U,E_D\cup E_D^{\leq (K+1)},\delta_I\cup\delta_D, v^U_0)$ with $V_U=V_{UA}\cup V_{UF}$.}
	\KwOut{$\mathcal{EM}=(V_E,E_D\cup E_D^{\leq (K+1)},\delta_I\cup\delta_D, v^E_0)$ with $V_E=V_{EA}\cup V_{EF}$.}
    Initialize $V_{PUA}=\emptyset$ and $V_{PUF}=\emptyset$\;
    \tcp{$V_{PUA}$ ($V_{PUF}$) is the set of states that need to be pruned away in $V_{UA}$ ($V_{UF}$)}
    \For{$v_{UF}\in V_{UF}$}{
		\For{$\omega\in E_D^{\leq (K+1)}$}{
			\If{$\exists z\in v_{UF}$, $\delta_D(z,\omega)$ is not defined}{
				Add $\delta_D(v_{UF},\omega)$ to $V_{PUA}$\;
			}    
		}
	}     	
    \For{$v_{UF}\in V_{UF}\setminus V_{PUF}$}{
			\If{$\forall \omega\in E_D^{\leq (K+1)}$, $\delta_D(v_{UF},\omega)\in V_{PUA}$}{
				Add $v_{UF}$ to $V_{PUF}$\;
			}
	}
	\For{$v_{UA}\in V_{UA}\setminus V_{PUA}$}{
			\If{$\exists\gamma\in E_D, \delta_D(v_{UA},\gamma)\in V_{PUF}$}{
				Add $v_{UA}$ to $V_{PUA}$\;
			}    
	}     	
	Return to line 6 until no new state that needs to be pruned away is produced\; 
	\For{$v_{UF}\in (V_{UF}\setminus V_{PUF})$}{
		\For{$\omega\in E_D^{\leq (K+1)}$}{
			\If{$\delta_{DI}(v_{UF},\omega)\in V_{PUA}$}{
				$C(v_{UF})=C(v_{UF})\setminus \omega$\;
			}
		}
	}
   Return $\mathcal{EM}= (V_E, E_D\cup E_D^{\leq (K+1)}, \delta, v_0^E)$.
   \tcp{$V_E=V_U\setminus (V_{PUA}\cup V_{PUF})$ and $v_0^E=v_0^U$}
\end{algorithm}

We next briefly explain how Algorithm \ref{alg_EMs2} works.
At lines~2 to 5, given a state $v_{UF}\in V_{UF}$ and an edit action $\omega\in E_D^{\leq (K+1)}$, we check if all information states in $v_{UF}$ perform the same edit action.
If there exists one state $z\in v_{UF}$ such that $\delta_D(z,\omega)$ is not defined, it means that the intruder can infer something is wrong since it may distinguish state $z$ among other states in $v$, which is not allowed.
Thus, we mark such states as those that need to be pruned away in $V_{PUA}$.
However, removing such states may violate the $i$-availability of the edit function such that it cannot respond to every system behavior.
To avoid this, we use an approach similar to Algorithm \ref{alg_TEM} so as to prune away states in a recursive manner in lines 6 to 16.

\begin{theorem}\label{theoremEM1}
	An edit function is $ic$-enforcing if and only if it can be synthesized from a non-empty edit mechanism $\mathcal{EM}$.   
\end{theorem}

\begin{proof}
    An edit function is $ic$-enforcing if and only if it is $i$-available, $c$-available, confidential, and integral.
    Since Lemma \ref{lemma1} established that an edit function is $i$-available and confidential if and only if it can be synthesized from the trimmed game structure, we only prove here that an edit function is $c$-available if and only if it can be synthesized from the edit mechanism.
    
	(If) By contradiction, we assume that the edit function is not $c$-available.
	Thus, it cannot ensure that every edited sequence can be recognized by the intruder; however, this is not allowed in $\mathcal{EM}$, which is a contradiction.
	Therefore, the edit function should be $c$-available.
 
	(Only if) If an edit function is $c$-available, it can be retained in $\mathcal{EM}$ since the transition function in $\mathcal{EM}$ is able to ensure that every edit action can be recognized by the intruder.
    Thus, it can be synthesized from the edit mechanism $\mathcal{EM}$.
\end{proof}

\begin{eg}
	Continuing with Example 5, the trimmed game structures can be transformed to the edit mechanism (shown in Fig. \ref{Editmechanism1}) as follows. 
    We first construct the no-guarantees edit mechanism by merging the states from the trimmed game structures.
    For instance, the initial state in the no-guarantees edit mechanism is $v_0^E=\lbrace v_0,v_1\rbrace$ in terms of $v^E_{0}= \lbrace v_0\rbrace\cup \lbrace v_1\mid\exists aa\in (E_o\setminus E_D)^*: \delta(v_0, aa)= v_1\rbrace$, where $v_0=(\lbrace1\rbrace,\lbrace1,4\rbrace,\lbrace1,3\rbrace)$ and $v_1=(\lbrace3\rbrace,\lbrace3,6\rbrace,\lbrace1,3\rbrace)$ in the trimmed game structure. 
    Then, the rest of the no-guarantees edit mechanism can be built by following Algorithm \ref{alg_EMs1}, as shown in Fig. \ref{Editmechanism1}.

\begin{figure}[htbp]
    \centering
    \includegraphics[width=0.5\textwidth]{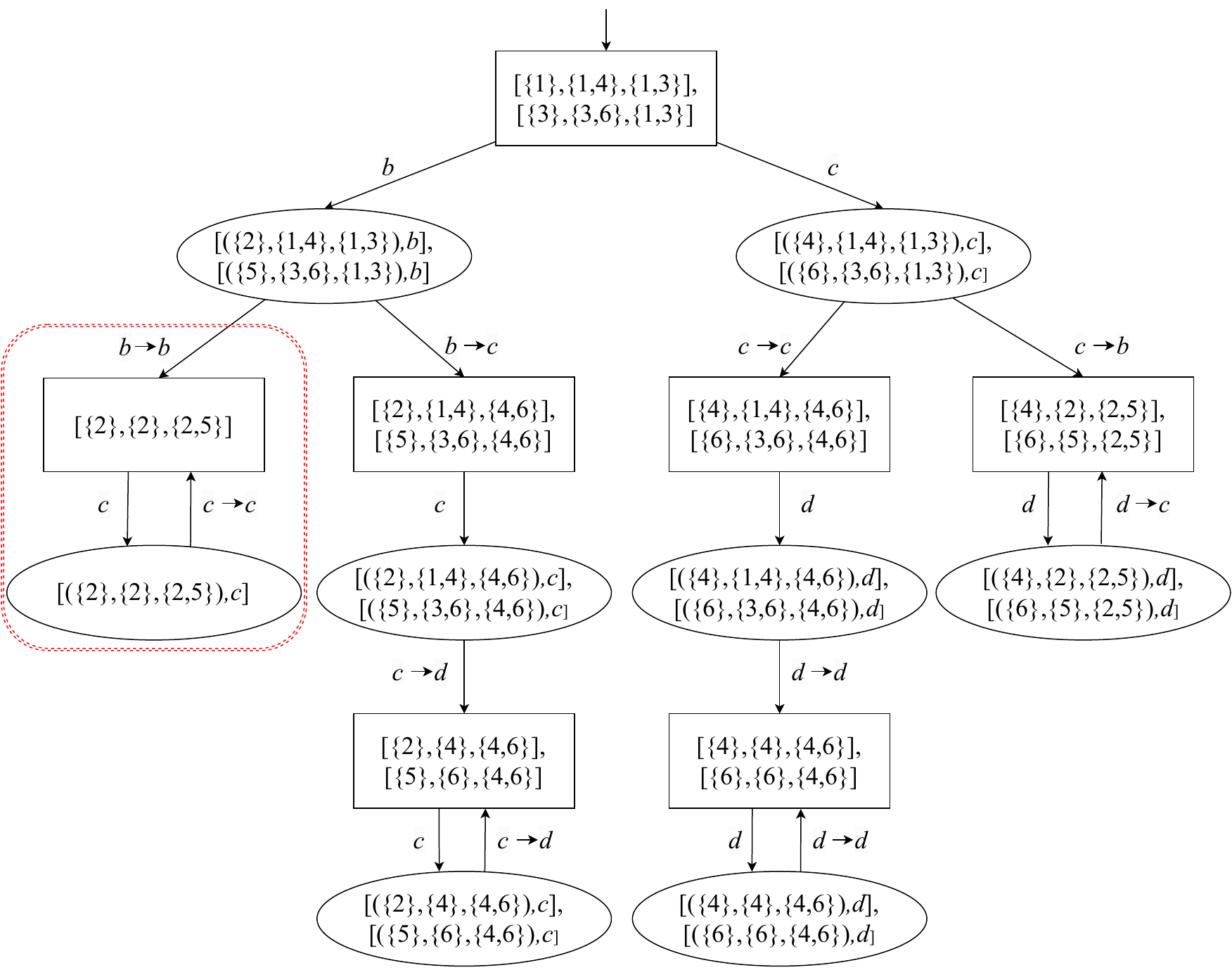}
    \caption{Edit mechanism.}
    \label{Editmechanism1}
\end{figure}

    By following Algorithm \ref{alg_EMs2}, we then construct the edit mechanism from the no-guarantees edit mechanism by pruning away states that do not implement the same edit action in a recursive manner. Consider state $v^E_1=\lbrace[(\lbrace2\rbrace,\lbrace1,4\rbrace,\lbrace1,3\rbrace),b],\allowbreak[(\lbrace5\rbrace,\allowbreak\lbrace3,6\rbrace,\lbrace1,3\rbrace),b]\rbrace$, which is reached from the initial state $v^E_{0}$ if the system executes event $b$ that can be observed by the edit function.
    At this point, one can see that state $v^E_2=\lbrace[(\lbrace2\rbrace,\lbrace2\rbrace,\lbrace2,5\rbrace)]\rbrace$ is reached if the edit function replaces $b$ with itself.
    However, this edit action is not defined at state $[(\lbrace5\rbrace,\lbrace3,6\rbrace,\lbrace1,3\rbrace),b]\in v^E_1$.
    Thus, we can conclude that the edit action that replaces $b$ with itself cannot be implemented on all states in $v^E_1$.
    Such states are removed by Algorithm \ref{alg_EMs2}, such that the resulting edit mechanism is obtained by removing the dotted box with red color in Fig.~\ref{Editmechanism1}.

    According to Theorem \ref{theoremEM1}, one concludes that the edit function synthesized from the edit mechanism in Fig. \ref{Editmechanism1} is $ic$-enforcing.
    At this point, the intruder will not be able to infer the secret state of the system.
    For instance, given a sequence $abc$ generated by the system, the intruder can infer the secret state $\lbrace 5\rbrace$ by following observation $P_I(abc)=ab$ via the system observer $\mathbf{O}_I$.
    The edit mechanism in Fig. \ref{Editmechanism1} reaches state $\lbrace [\lbrace2\rbrace,\lbrace4\rbrace,\lbrace4,6\rbrace], [\lbrace5\rbrace,\lbrace6\rbrace,\lbrace4,6\rbrace]\rbrace$ via its observation $P_D(abc)=bc$ and corresponding edit actions (i.e., replacing $b$ with $c$ and replacing $c$ with $d$).
    At this point, the intruder's estimation may be state $\lbrace4\rbrace$ or $\lbrace6\rbrace$ after receiving the modified observations, which means that it cannot infer the secret.
	$\hfill\diamond$
\end{eg}

\begin{remark}
  Due to the fact that the edit mechanism contains all $iu$-enforcing edit functions, one can always synthesize an edit function from the edit mechanism such that it can enforce opacity of the system.
  Recall that an edit function is designed to interface at the output of a system. That is, it receives the output of the system as its input and modifies it to present a different output to the intruder. 
  Thus, we are able to construct an abstract discrete model (i.e., Mealy machine) to encode the edit function.  
  The reader is referred to Algorithm 5 in \cite{2014Wu} for more details on the synthesizing procedures.
  $\hfill\diamond$
\end{remark}

\section{Discussion: complexity analysis and comparison with existing works}
\label{sec:discussion}

In this section, we first discuss the computational complexity of an edit mechanism. 
Then, we discuss the relationship of our work with existing ones. 

\subsection{Complexity Analysis}

The main structure underlying the problem considered in this paper is an edit mechanism, which is constructed through five steps (see Fig. \ref{Constructionflow}).
To understand the complexity of each step, we calculate the worst-case space complexity of each structure in terms of the previous structure.

\begin{table*}[htbp]
\centering
\begin{tabular}{|c|c|c|c|c|c|c|}
\hline
Works                                    & \cite{2010Dubreil, 2015Yin}    &  \cite{2018Tong} & \cite{2014Wu, 2018Wu} &  \cite{2018Ji, 2019Ji}  &  \cite{2023Duan} & This paper  \\ \hline
Observation capability relation       & $E_I\subseteq E_D$ & No inclusion relation & $E_I=E_D$ & $E_I=E_D$ & $E_D\subseteq E_I$ & No inclusion relation  \\ \hline
The intruder's awareness of the defender &  Yes  &  No &  No  &  Yes  & No &  No \\ \hline
Restriction of the system behavior       &  Yes   &  Yes   &  No & No  & No & No \\ \hline
\end{tabular}
   \caption{Classification of existing approaches for opacity enforcement.}
    \label{tab:my_label}
\end{table*}

Given a DES modeled by a DFA $G=(X,E,f,x_0)$ with $E=E_o\cup E_{uo}$, the sets of events observed by the intruder and the edit function are $E_I\subseteq E_o$ and $E_D\subseteq E_o$, respectively. 

\begin{enumerate}
    \item At step 1: We build the system, intruder, and the defender observers.
    By taking into account the number of states $|X|$ in $G$, $E_I\subseteq E_o$, and $E_D\subseteq E_o$, the system, intruder, and the defender observers have $|X_o|$, $|X_I|$, and $|X_D|$ states, all bounded by $2^{|X|}$.
    \item At step 2: We build an edit game structure $\mathcal{EGS}$, which contains two types of states.
    Thus, the number of states in $\mathcal{EGS}$ is at most $|V|=|V_A|+|V_F|=|X_o||X_I||X_D|+|X_o||X_I||X_D||E_o|\subseteq 2^{3|X|}+2^{3|X|}|E_o|$.
    \item At step 3: We build a trimmed game structure $\mathcal{TGS}$.
    We can say that the number of states in $\mathcal{TGS}$ is at most $|TV|=|V|\subseteq 2^{3|X|}+2^{3|X|}|E_o|$ if there exists no problematic state.
    \item At step 4: We build a no-guarantees edit mechanism $\mathcal{UEM}$.
    Since it is constructed by merging states from $\mathcal{TGS}$, the number of states in $\mathcal{UEM}$ is at most $|V_U|=|V_{UA}|+|V_{UF}|\subseteq 2^{|V_{TA}|}+2^{|V_{TF}|}\subseteq 2^{2^{3|X|}}+2^{2^{3|X|}|E_o|}|E_D|$. 
    \item At step 5: We build an edit mechanism $\mathcal{EM}$. 
    We can say that the number of states in $\mathcal{EM}$ is at most $|V_E|=|V_U|\subseteq 2^{2^{3|X|}}+2^{2^{3|X|}|E_o|}|E_D|$ if there exists no problematic state. 
    Overall, the structural complexity of the edit mechanism is $O(2^{2^{3|X|}(1+|E_o|)})$.
\end{enumerate}

Specifically, the synthesis of our edit mechanism proceeds via two nested powerset constructions (first for the intruder/defender observers, then for the imperfect information game), each of which incurs an exponential growth in the size of the previous structure.
However, solving two-player games with imperfect information is known to be 2-EXPTIME-complete in the general case \cite{chatterjee2010complexity}. 
Thus, the overall worst-case complexity of our construction is doubly-exponential in the number of system states.
In practice, some abstraction techniques \cite{mayr2013advanced} can often reduce the effective state space before game synthesis, yielding tractable instances for realistic models.

\subsection{Comparison with Existing Work}

The problem of opacity enforcement has been extensively studied in the literature, following two primary approaches: supervisor control \cite{2010Dubreil, 2012Chrisa, 2015Yin, 2018Tong} and obfuscation techniques \cite{2014Wu,2016Wu, 2018Wu, 2018Ji, 2019Ji, 2023Duan}. These approaches involve designing an external defender, which can be a controller, an insertion function, or an edit function, to confront the intruder based on various practical assumptions. 
In the following, we provide a comparison between our work and the existing studies by categorizing these assumptions into three key aspects.
For the sake of consistency, we denote the sets of events observed by the intruder and the defender (such as the controller or the insertion/edit function) as $E_I$ and $E_D$, respectively.

\begin{enumerate}
    \item Observation capabilities: The common assumption regarding observation capabilities in opacity enforcement problems is that both the intruder and the defender possess partial but identical observation capabilities of the system, i.e., $E_I=E_D$.
    A challenging scenario arises when there is no inclusion relation between the sets of events observed by the intruder and the defender.
    This situation is termed the ``incomparable observation setting", where $E_I$ and $E_D$ have no inclusion relation.
    Tong $\emph{et al.}$ \cite{2018Tong} have adopted this assumption by using supervisor control theory. 
    However, we deal with it by employing edit functions, which is challenging since we also need to ensure that the intruder can recognize the outputs of the edit function under this incomparable observation setting.
    \item Knowledge of the intruder: 
    The concept of intruder knowledge in opacity enforcement problems using obfuscation techniques explores various scenarios in which the intruder's level of awareness regarding the defender differs. 
    We can categorize these scenarios into two distinct cases: one where the intruder is not aware of the defender's existence (see, for example, \cite{2014Wu, 2018Wu, 2023Duan}), and the other where the intruder is aware of the defender's existence (see, for example, \cite{2018Ji,2019Ji}). 
    In this work, we extensively investigate the first case, whereas the second case will remain open for us to extend under the incomparable observation setting.
    \item Restriction of the system behavior: The approach by using supervisory control theory \cite{2010Dubreil, 2012Chrisa, 2015Yin, 2018Tong} typically focuses on designing a controller that restricts the system's behavior to ensure opacity.
    In contrast, obfuscation techniques \cite{2014Wu,2016Wu, 2018Wu, 2018Ji, 2019Ji, 2023Duan}, including our work, emphasize the use of external mechanisms (i.e., insertion or edit functions) to manipulate observations generated by the system, instead of restricting them.
\end{enumerate}

The comparison against existing work is summarized in Table \ref{tab:my_label}.
In conclusion, our work considers the problem of opacity enforcement via obfuscation techniques (i.e., edit functions), specifically addressing scenarios with incomparable observation capabilities and the case that the intruder is not aware of the defender's presence. 

\section{Conclusions}\label{sec:conclusion}
In this paper, we investigate the problem of opacity enforcement via edit functions within a more complex and general setting, where the edit function and the intruder possess observation capabilities that are incomparable. 
To deal with this issue, the concept of $ic$-enforceability is introduced to provide a systematic characterization of the ability of edit functions to enforce opacity of a system.
Subsequently, a comprehensive edit mechanism is constructed within a game-theoretical scheme under imperfect information and is shown to encompass all edit functions that are $ic$-enforcing.
Our future work will focus on i) extending the work in the cases where the intruder is aware of the edit functions; ii) reducing the computational complexity of the edit mechanisms, thereby advancing the practicality and scalability of the proposed approaches.

\bibliographystyle{IEEEtran}
\bibliography{Ref}

\end{document}